\documentclass[format=acmsmall,review=false,nonacm]{acmart}
\usepackage{booktabs} 
\usepackage{amsthm}
\usepackage{amsmath}
\usepackage{mathtools}
\usepackage{xspace}
\usepackage[f]{esvect}
\usepackage{bm}

\usepackage{natbib}
\usepackage[capitalize]{cleveref}
\usepackage{algorithm}
\usepackage[noend]{algpseudocode}

\usepackage{multicol}
\usepackage[bottom]{footmisc}

\theoremstyle{plain}
\newtheorem{theorem}{Theorem}
\newtheorem{proposition}{Proposition}
\newtheorem{lemma}{Lemma}
\newtheorem{corollary}{Corollary}

\theoremstyle{definition}
\newtheorem{definition}{Definition}
\newtheorem{example}{Example}

\newcommand{\election}{$\elec = (V, C, \prof)$\xspace}

\newcommand{\elec}{\mathcal{E}}
\newcommand{\prof}{{\vv{\succ}}}
\renewcommand{\succeq}{\succcurlyeq}
\renewcommand{\top}{\mathsf{top}}
\renewcommand{\bot}{\mathsf{bot}}
\newcommand{\invbot}{\bot\inv}
\newcommand{\plu}{\mathsf{plu}}
\newcommand{\veto}{\mathsf{veto}}
\newcommand{\uni}{\mathsf{uni}}
\newcommand{\weight}{\mathsf{weight}}
\newcommand{\dom}{\mathcal{D}}
\newcommand{\cost}{\mathsf{cost}}
\newcommand{\dist}{\mathsf{dist}}
\newcommand{\neigh}{\mathcal{N}}
\newcommand{\timer}{\mathsf{time}}

\newcommand{\mvec}{\ensuremath{\mathbf{M}}\xspace}
\newcommand{\mrow}[1]{\mvec(#1)}
\newcommand{\m}[1]{M(#1)}

\newcommand{\mhatvec}{\ensuremath{\widehat{\mathbf{M}}}}
\newcommand{\mhat}[1]{\widehat{M}(#1)}

\newcommand{\mbarvec}{\ensuremath{\widebar{\mathbf{M}}}}
\newcommand{\mbar}[1]{\widebar{M}(#1)}

\newcommand{\xvec}{\ensuremath{\mathbf{x}}\xspace}

\newcommand{\pvec}{\ensuremath{\mathbf{p}}\xspace}
\newcommand{\p}[1]{p(#1)}
\newcommand{\qvec}{\ensuremath{\mathbf{q}}\xspace}
\newcommand{\q}[1]{q(#1)}
\newcommand{\pq}{\ensuremath{{(\pvec, \qvec)}}\xspace}

\newcommand{\sequentialrule}{\textsc{SerialVeto}\xspace}
\newcommand{\anonrule}{\textsc{SimultaneousVeto}\xspace}
\newcommand{\anonplurule}{\textsc{SimultaneousPluralityVeto}\xspace}
\newcommand{\pluralityveto}{\textsc{PluralityVeto}\xspace}
\newcommand{\pluralitymatching}{\textsc{PluralityMatching}\xspace}
\newcommand{\pluralityrule}{\textsc{Plurality}\xspace}
\newcommand{\vetorule}{\textsc{Veto}\xspace}

\newcommand{\uc}{prefix-intersecting\xspace}
\newcommand{\axiom}[1]{\textbf{#1}}

\DeclareMathOperator{\argmin}{argmin}

\providecommand{\Set}[2]{\ensuremath{\SET{#1 \mid #2}}\xspace}
\providecommand{\SET}[1]{\ensuremath{\{ #1 \}}\xspace}
\providecommand{\SetCard}[1]{\ensuremath{| #1 |}\xspace}

\DeclareMathSymbol{\R}{\mathord}{AMSb}{"52}
\newcommand{\N}{\mathbb{N}}
\newcommand{\inv}{^{-1}}
\newcommand{\comp}[1]{\overline{#1}}
\DeclarePairedDelimiter{\ceil}{\lceil}{\rceil}
\newcommand{\dd}{\ensuremath{\mathrm{d}}\xspace}
\newcommand{\ind}{\mathcal{I}}

\title{Generalized Veto Core and a Practical Voting Rule with Optimal Metric Distortion}
  
\author{Fatih Erdem Kizilkaya}
\email{fatih.erdem.kizilkaya@gmail.com}
\orcid{0000-0001-5786-1184}
\author{David Kempe}
\orcid{0000-0003-4002-9759}
\email{david.m.kempe@gmail.com}
\affiliation{%
  \institution{University of Southern California}
  \country{USA}
}

\begin{abstract}
We revisit the recent breakthrough result of Gkatzelis et al. on (single-winner) metric voting, which showed that the optimal distortion of 3 can be achieved by a mechanism called \textsc{PluralityMatching}.
The rule picks an arbitrary candidate for whom a certain candidate-specific bipartite graph contains a perfect matching.
Subsequently, a much simpler rule called \textsc{PluralityVeto} was shown to achieve distortion 3 as well;
this rule only constructs such a matching implicitly, but it, too, makes some arbitrary decisions affecting the outcome.
Our point of departure is the question whether there is an intuitive interpretation of this matching, with the goal of identifying the underlying source of arbitrariness in these rules.
We first observe directly from Hall's condition that a matching for candidate $c$ certifies that there is no coalition of voters that can jointly counterbalance the number of first-place votes $c$ received, along with the first-place votes of all candidates ranked lower than $c$ by any voter in this coalition.
This condition closely mirrors the classical definition of the (proportional) veto core in social choice theory, except that coalitions can veto a candidate $c$ whenever their size exceeds the plurality score of $c$, rather than the average number of voters per candidate.
Based on this connection, we define a general notion of the veto core with arbitrary weights for voters and candidates which respectively represent the veto power and the public support they have.

This connection opens up a number of immediate consequences.
Previous methods for electing a candidate from the veto core can be interpreted simply as matching algorithms.
Different election methods realize different matchings, in turn witnessing different sets of candidates as winners.
Viewed through this lens, we first resolve nontrivial tie breaking issues contributing to the inherent arbitrariness of the above rules.
Our approach to ties reveals a novel characterization of the (general) veto core, showing it to be identical to the set of candidates who can emerge as winners under a natural class of matching algorithms reminiscent of \textsc{SerialDictatorship}.
Then, we extend this class of voting rules into continuous time, and obtain a highly practical voting rule with optimal distortion 3, which is also intuitive and easy to explain:
Each candidate starts off with public support equal to his plurality score. 
From time 0 to 1, every voter continuously brings down, at rate 1, the support of her bottom choice among not-yet-eliminated candidates.
A candidate is eliminated if he is opposed by a voter after his support reaches 0.
On top of being anonymous and neutral, the absence of arbitrary non-deterministic choices in this rule allows for the study of other axiomatic properties that are desirable in practice.
We show that the canonical voting rule we propose for electing from the (plurality) veto core also satisfies  
resolvability, monotonicity, majority, majority loser, mutual majority and reversal symmetry, in addition to still guaranteeing metric distortion 3.

\end{abstract}



\begin{document}

\maketitle

\section{Introduction}
\label{sec:intro}
Voting is a fundamental process by which a group of agents with different preferences can decide on one alternative; applications range widely, from groups of friends choosing a restaurant or meeting time, to departments and other bodies deciding on a course of action or budget, to towns, states or countries electing politicians.
When each voter can support or oppose only one candidate, the~most informative votes are those for their top or bottom choices.
As a result, \pluralityrule and \vetorule \ --- respectively, picking the candidate with the largest and smallest number of such votes --- are the most natural voting rules for single-vote elections.
Better results can be obtained when the~voters' full rankings are taken into account, and indeed, a large body of theoretical and practical work has focused on \emph{ranked-choice} voting rules, which have these complete rankings available.
Of~course, there are many natural ways to aggregate this information into the choice of a winner. 
The choice of voting rule, ultimately, depends on the specific context and goals of the election.
When picking a restaurant or meeting time within a group of friends, the preferences of minorities, such as those with dietary restrictions or limited availability, may be more important; thus, using \vetorule might be a better idea than \pluralityrule.
Conversely, for high-stakes settings, such as presidential elections, picking a winner opposed by a majority could lead to conflicts; thus, the benefits of \pluralityrule may outweigh those of \vetorule.

In light of the plethora of available (ranked-choice) voting rules, it becomes important to evaluate and compare the pros and cons of different rules for different settings.
One approach to do so, dating back to at least the 18th century \citep{borda:elections,condorcet:essay}, is via axioms which the rule may or may not satisfy; see \citet{BCULP:social-choice} for a detailed overview.
Among those, two absolutely essential axioms are \axiom{Anonymity} and \axiom{Neutrality} which, respectively, require that all voters and candidates be treated the same a priori.
As shown by \citet{moulinStrategySocialChoice1983} (though similar ideas are present in earlier work, e.g., \citet{smith:variable-electorate}), no voting rule can be both anonymous and neutral without having tied outcomes.
In order to deal with the presence of ties necessitated by \axiom{Anonymity} and \axiom{Neutrality}, another essential axiom is \axiom{Resolvability} \citep{tidemanIndependenceClonesCriterion1987} which requires that ties can always be broken in favor of any winning candidate $w$ by adding a single voter who ranks $w$ at the top.

An alternative approach to defining desirable axiomatic properties, proposed much more recently, and to date primarily pursued by the computer science community, is to interpret voting as an~implicit \emph{optimization} problem of finding a ``best'' candidate.
Here, a natural modeling assumption is that the voters and candidates are jointly embedded in an abstract metric space \citep{anshelevich:bhardwaj:postl,anshelevich:bhardwaj:elkind:postl:skowron,anshelevich:ordinal,anshelevich:filos-ratsikas:shah:voudouris:retrospective,anshelevich:filos-ratsikas:shah:voudouris:reading-list}, and the distance between a voter and a candidate captures how much the voter dislikes the candidate, or how much their opinions/positions on key issues differ.
This generalizes the notion of \emph{single-peaked} preferences \citep{black:rationale,downs:democracy,moulin:single-peak}, which in its original definition only considered the line instead of a general metric space.
A natural definition of a~``best'' candidate is then one who minimizes the total distance to all voters, and a good voting rule should approximately minimize this total distance with its choices of candidates.
However, there is a~key obstacle: the~voting rule does not have access to the actual distances\footnote{Indeed, it would be very unlikely that voters could even articulate the nature of the ``opinion space,'' let alone estimate/communicate precise distances in it.}. 
Instead, the rankings provide \emph{ordinal} information about the distances, in that each voter ranks candidates from closest to furthest.\footnote{It has been observed repeatedly that humans are much better at \emph{comparing} options than articulating the desirability of an option in isolation.} 
This lack of information about the cardinal distances can of course result in suboptimal choices for any rule, and the worst-case multiplicative gap between the sum of distances from the chosen candidate to all voters, vs. the optimum one chosen with knowledge of the metric space, has been termed \emph{distortion} of the voting rule. (See \cref{sec:distortion} for the formal definition.)

\citet{anshelevich:bhardwaj:postl} established a lower bound of 3 on the distortion of any deterministic voting rule.
After a sequence of improvements and partial progress \citep{munagala:wang:improved,DistortionDuality}, this lower bound was matched by the voting rule \pluralitymatching in a~recent breakthrough result \citep{gkatzelis:halpern:shah:resolving} which showed that the voting rule has distortion 3.
\pluralitymatching is a complex rule not likely to be understood by anyone without mathematical background.
This shortcoming was remedied in \citet{PluralityVeto}, where it was shown that a very simple rule called \pluralityveto enjoys the same optimal distortion 3 guarantee as \pluralitymatching.
\pluralityveto works as follows.
Each of the $m$ candidates starts with a~score equal to his\footnote{For consistency, we will always refer to voters with female and candidates with male pronouns.} number of first-place votes.
Then, each of the $n$ voters is processed in an~arbitrary order.
When it is the turn of a voter, she considers all candidates who currently have positive scores, and subtracts 1 from the candidate she likes least among them.
Note that candidates start with a~total score of $n$, and because each voter subtracts 1, all candidates finish with a score of~0.
The~last candidate whose score reaches 0 is chosen as the winner.

The simplicity of \pluralityveto, along with its natural intuition of counterbalancing top votes against bottom votes, suggests that it might be suitable for implementation in practice. 
However, \pluralityveto suffers from a fatal flaw which precludes its acceptance in practice: it violates the~\axiom{Anonymity} axiom. 
This can be seen most immediately in the following basic example.

\begin{example} \label{ex:obvious-tie}
  There are two candidates $\SET{a,b}$ and two voters $\SET{1,2}$ with preferences $a \succ_{1} b$ and $b \succ_{2} a$.
  If voter 1 vetos first, $a$ wins --- otherwise, $b$ wins.
  Stated differently: if the veto order is $(1, 2)$, then switching the voters' preferences changes the outcome, violating \axiom{Anonymity}.\footnote{The \axiom{Anonymity} property could be restored by choosing a ranking-dependent veto order, e.g., having voters veto in lexicographic order of their ballots. However, such a solution would come at the cost of violating the equally important \axiom{Neutrality} axiom.}
\end{example}

Naturally, this example poses a problem for \emph{every} deterministic voting rule.
The election from \cref{ex:obvious-tie} should clearly be declared a tie, and resolved separately by some tie breaker, such as a~coin flip.
The following question then becomes central: 
if a~mechanism akin to \pluralityveto is to return a tied subset of candidates, all of whom would guarantee metric distortion\footnote{Akin to the distortion of a mechanism, the distortion of a candidate is the ratio between his sum of distances to all voters, and the sum of distances of an optimal candidate.} 3, how, and on what grounds, should such a subset be chosen?

\medskip

To begin answering this question, and develop practical mechanisms of (metric) distortion 3, we begin by revisiting the \pluralitymatching mechanism of \citet{gkatzelis:halpern:shah:resolving}. 
\pluralitymatching is based on the notion of \emph{domination graphs}: the domination graph $G_c$ for candidate $c$ is a bipartite graph between voters and voters in which an edge from $v$ to $v'$ exists if and only if $v$ weakly prefers $c$ over the top choice of $v'$. (See \cref{sec:domination-graphs} for the formal definition.)
It~had been shown by \citet{munagala:wang:improved} and \citet{DistortionDuality} that any candidate $c$ for whom $G_c$ contains a perfect matching has distortion 3; the breakthrough of \citet{gkatzelis:halpern:shah:resolving} was to show (non-constructively) that such a candidate $c$ always exists.
Their voting rule \pluralitymatching finds such a candidate by brute force.
This description is non-deterministic about which candidate(s) will be returned in case there are multiple such $c$. Picking a candidate deterministically would violate the \axiom{Neutrality} axiom. On the other hand, the following example shows that returning all such candidates violates \axiom{Resolvability}.

\begin{example}\label{ex:resolvability}
There are three candidates $\SET{a,b,c}$, and two voters $\SET{1, 2}$ with rankings $a \succ_1 b \succ_1 c$ and $c \succ_2 b \succ_2 a$.
All candidates have the perfect matching $\SET{(1, 2), (2, 1)}$ in their domination graph, and thus would be winners under \pluralitymatching returning all candidates with perfect matchings.
Even if a voter is added who ranks $b$ at the top (and without loss of generality $a$ second), both $G_b$ and $G_a$ contain perfect matchings, so the election cannot be resolved in favor of $b$.
\end{example}

The execution of \pluralityveto can be viewed as not just determining a winning candidate $c$, but constructing a~perfect matching in $G_c$ along the way --- thus, explicitly \emph{witnessing} distortion 3.
Our endeavor to improve on \pluralityveto, and understand \pluralitymatching more deeply, begins with an examination of the role of these matchings.
Besides witnessing distortion 3, is there any constructive interpretation that one could explain to a voter not trained in graph theory, but interested in acquiring intuition for the voting rule?
The following observation follows directly from Hall's Condition:
The graph $G_c$ contains a perfect matching if and only if there is no subset $S$ of voters whose size is larger than the total number of first-place votes of not only $c$, but also all candidates whom at least one member of $S$ ranks lower than $c$.
In other words, no coalition of voters has enough combined veto power to counterbalance the~first-place votes of $c$ and candidates ranked lower by the coalition.

This condition is highly reminiscent of the classic definition of the \emph{veto core} \citep{moulinProportionalVetoPrinciple1981}; the latter simply has \emph{uniform} scores of $n/m$ for candidates which need to be counterbalanced, as opposed to the numbers of first-place votes. (See \cref{sec:proportional} for the formal definition.)
Thus, we~define a generalized notion of the veto core, in which voters and candidates can have arbitrary weights, so long as the sums of the weights are the same.
It then follows immediately from Hall's condition that the veto core consists of the~candidates $c$ for whom a natural weighted generalization of the bipartite graph $G_c$ (defined now between voters and candidates) has a (fractional) perfect matching (see \cref{sec:equivalance}).

Past algorithms for electing a candidate from the veto core \citep{moulinVotingProportionalVeto1982} can then simply be understood as constructing a \emph{witnessing (perfect) matching} for some candidate (though various subtle tie-breaking choices in earlier definitions introduce some strange artifacts).
The witnessing matchings possess rich structure, and focusing on them turns out to be very fruitful in many respects. 
First, notice that the same perfect matching $\mvec$ may exist in the~domination graphs $G_c$ of several distinct candidates $c$; let $W(\mvec)$ be the set of such candidates. 
By focusing on the construction of the~witnessing matching $\mvec$ as the primary task of an algorithm, $W(\mvec)$ emerges as a natural set of tied winners to return. 
Interestingly, $W(\mvec)$ is convex in all possible embedding spaces for all~$\mvec$, providing a~desirable ``consistency'' condition for the returned set of tied winners (see \cref{sec:convex}).

Both past work on the veto core and (the analysis of) \pluralityveto can be viewed as a~greedy construction of witnessing matchings.
We show that these algorithms are special cases of a general class of natural algorithms which greedily construct a witnessing matching for arbitrary weighted veto cores.
The fact that \pluralityveto only selects a single winner is an artifact of an~overly strict elimination rule. 
By making some important modifications --- specifically, eliminating a candidate not when his score reaches 0 but only \emph{if he is subsequently opposed by some voter} --- we obtain the~desirable property that whenever a witnessing matching is produced, \emph{the entire set $W(\mvec)$} is returned as tied winners (see \cref{sec:serialveto}).
This also resolves another undesirable property of \pluralityveto: that a candidate $c$ who is ranked second by every voter could never win, even when $G_c$ contains a perfect matching.
In addition, we show that for arbitrary weights on candidates and voters, the~corresponding veto core exactly equals the set of candidates who can be returned as winners by the modified algorithm; the proof of this result utilizes an argument highly reminiscent of the~\textsc{TopTradingCycles} algorithm (see \cref{sec:characterization}).

\subsubsection*{Our Main Contribution}
Returning to our original goal, we present an anonymous, neutral and resolvable voting rule that achieves optimal metric distortion by implicitly constructing a witnessing matching $\mvec$ and returning $W(\mvec)$.
Here, we adapt an algorithm of \citet{ianovskiComputingProportionalVeto2021} for the veto core, which can be construed as vetoing candidates continuously over time.
As in \pluralityveto, each candidate starts with a score equal to his number of first-place votes; however, our voting rule starts off with the set of all candidates as eligible winners (even those with score 0). 
Then, all voters are processed simultaneously in continuous time from time 0 to 1. 
Every voter, starting from her bottom choice, \emph{instantly} eliminates those with score 0 and moves on to her next bottom choice.
When she arrives at a candidate with positive score, she decrements his score continuously at rate~1.
The set of winners consists of all candidates who are not eliminated until time 1 (see \cref{sec:anonrule}).

In addition to guaranteeing optimal distortion (for all tied winners), and satisfying the essential axioms of \axiom{Anonymity}, \axiom{Neutrality} and \axiom{Resolvability}, our practical rule satisfies many other basic axioms: \axiom{Monotonicity}, \axiom{Majority}, \axiom{Majority Loser}, \axiom{Mutual Majority}, and \axiom{Reversal Symmetry} (see \cref{sec:axioms} for definitions), and it is also easy to explain, even without reference to distortion.
Thus, we believe that the rule is indeed practical.

Our work unites not only two bodies of work on voting but also two differentiating factors in voting settings: the protection of the majority versus minorities. 
The motivation for the notion of veto core is the protection of minorities (and it is optimal in that sense under some formalizations of this principle \citep{kondratevMeasuringMajorityPower2020}),
whereas optimizing metric distortion clearly aims for the protection of the majority, by (approximately) minimizing the (utilitarian) social cost. 
Thus, the~generalized veto core gives rise to a spectrum of voting rules via initial weights interpolating between the plurality scores and uniform weights, and consequently, between the protection of majority and minorities.
For further discussion and other possible future directions, see \cref{sec:conclusion}.

\subsubsection*{Related Work}
\label{sec:relatedwork}
\noindent Selection methods in which parties alternately cross off alternatives in each round are common in practice and have been a subject of research in various settings, such as arbitration schemes \cite{bloomAnalysisSelectionArbitrators1986, anbarciFiniteAlternatingMoveArbitration2006, barberaCompromisingCompromiseRules2022} and negotiation schemes \cite{erlichNegotiationStrategiesAgents2018, anbarciNoncooperativeFoundationsArea1993, barberaBalancingPowerAppoint2017}.
In the context of voting, a~sequential veto-based mechanism was perhaps first studied formally by \citet{muellerVotingVeto1978}, but only for alternatives consisting of one proposal from each voter, in addition to the status quo.  
In the analysis of \citet{muellerVotingVeto1978}, voters take turns in an arbitrary order and strike off the alternative they like the least. 
Exactly one outcome is selected by this rule, and clearly no voter will see their worst outcome elected. 
However, this~rule requires the number of alternatives to be one greater than the number of voters.
Much of the follow-up work has focused on strategic behavior by voters under this type of rule.
For example,\citet{pelegConsistentVotingSystems1978,moulinImplementingEfficientAnonymous1980,moulinPrudenceSophisticationVoting1981} show that variants of iterated elimination rules can elect the same winner under strategic behavior as with truthful behavior, for various equilibrium notions.

\citet{moulinProportionalVetoPrinciple1981} extended the concept of voting by veto from individuals to coalitions to study the core of the resulting cooperative game, and thus, the~proportional veto core was introduced.
Subsequently, \citet{moulinVotingProportionalVeto1982} proposed a rule electing from the proportional veto core, and studied the strategic behavior by voters under this rule.
Recently, \citet{ianovskiComputingProportionalVeto2021} (see also the expanded version \citet{ianovskiComputingProportionalVeto2023}) showed how to compute the~proportional veto core in polynomial time and introduced an anonymous and neutral rule electing from the proportional veto core.
To some degree, the veto core is also related to the notion of proportional fairness, which measures whether the influence of cohesive groups of agents on the voting outcome is proportional to the group size \cite{ebadianOptimizedDistortionProportional2022}.

The utilitarian analysis of voting rules through the lens of approximation algorithms was first suggested by \citet{BCHLPS:utilitarian:distortion,caragiannis:procaccia:voting,procaccia:approximation:gibbard,procaccia:rosenschein:distortion}.
\citet{boutilier:rosenschein:incomplete,anshelevich:bhardwaj:postl} were the first to clearly articulate the tension between the objective of maximizing utility (or minimizing cost) and the available information, which is only ordinal; they also termed the resulting gap \emph{distortion}.
In the earlier work, the focus was on (positive) utilities, and no additional assumptions (such as metric costs) were placed on the utilities \cite{BCHLPS:utilitarian:distortion,caragiannis:procaccia:voting,procaccia:approximation:gibbard,procaccia:rosenschein:distortion}.
Without any structures on the costs, the lowest possible distortion is $O(m^2)$ \cite{caragiannisSubsetSelectionImplicit2017}.
Connections between both types of distortion analysis have also been made: Very recently, \citet{BestOfBothDistortionWorlds} showed that a variant of \pluralityveto achieves the asymptotically optimal distortion of $\Theta(1)$ and $\Theta(m^2)$, respectively, under the metric and utilitarian notions of distortion.

There has been extensive work on giving lower and upper bounds on the metric distortion of specific rules in the social choice literature \cite{anshelevich:bhardwaj:postl, anshelevich:bhardwaj:elkind:postl:skowron, skowronSocialChoiceMetric2017, goel:krishnaswamy:munagala, munagala:wang:improved, DistortionDuality}.
The notion of metric distortion has also  been studied for settings other then single-winner elections, such as 
committee election \cite{caragiannisMetricDistortionMultiwinner2022}, 
facility location \cite{anshelevichOrdinalApproximationSocial2021},
distrubuted voting \cite{voudourisTightDistortionBounds2023, anshelevichDistortionDistributedMetric2022}
and distributed facility location \cite{filos-ratsikasSettlingDistortionDistributed2023, filos-ratsikasApproximateMechanismDesign2021}.
We refer the reader to the surveys of \citet{anshelevich:filos-ratsikas:shah:voudouris:retrospective, anshelevich:filos-ratsikas:shah:voudouris:reading-list} for further discussion of related work on metric distortion. 

In this work, we focus exclusively on \emph{deterministic} voting rules. The distortion --- both utilitarian and metric --- of randomzied voting rules has also been studied extensively.
\citet{anshelevich:postl:randomized} showed that \textsc{Random Dictatorship} --- the rule electing the top choice of a uniformly random voter --- achieves expected distortion $3-\frac{2}{n}$.
This bound is strictly better than the lower bound of 3 for deterministic rules.
The bound was improved slightly by \citet{DistortionCommunication} to $3-\frac{2}{m}$, by mixing between \textsc{Random Dictatorship} and another simple rule.
In the setting of our \cref{ex:obvious-tie}, it is easy to show that no randomized rule can achieve expected distortion less than 2.
Perhaps somewhat surprisingly, this \emph{lower} bound was improved to 2.1126 by \citet{charikar:ramakrishnan:randomized-distortion}.
The~gap between the upper and lower bounds constitutes one of the most intriguing open questions about the distortion of voting rules.
While there is no concrete direction, it is conceivable that an in-depth study of the generalized veto core may lead to useful insights towards closing or reducing the gap.

\section{Preliminaries}
\label{sec:prelim}
Throughout, vectors and matrices are denoted by bold lowercase and uppercase letters, respectively.
We denote the $i^{\text{th}}$ element of a vector $\xvec$ by $x(i)$,
and we extend this notation to sets via addition, i.e., $x(T) = \sum_{i \in T} x(i)$.
We denote the $i^{\text{th}}$ row vector of a matrix $\mvec$ by $\mrow{i}$ and we denote the~element in the $i^\text{th}$ row and $j^\text{th}$ column by $\m{i, j}$.
Given a set $S$, let $\Delta(S)$ denote the probability simplex over $S$, i.e., the set of non-negative weight vectors over $S$ that add up to 1.
Given a weight vector $\xvec$ defined over a set $S$, we use $[\xvec \neq 0]$ to denote the elements of $S$ with non-zero weights, i.e., $[\xvec \neq 0] = \Set{i \in S}{x(i) \neq 0}$. 
Given integers $i$~and~$j$, we use $[i, j]$ to denote the set of integers between $i$ and $j$ (inclusively).  

\bigskip

Let $V$ be the set of $n$ \emph{voters}, and let $C$ be the set of $m$ \emph{candidates}. 
For each voter $v$, let $\succ_v$ be a~\emph{ranking} (i.e., total order) over $C$.
The vector of all rankings $\prof = (\succ_v)_{v \in V}$ is referred to as a~\emph{ranked-choice profile}, and \election is referred to as an \emph{election}.
A \emph{voting rule} $f$ is an algorithm that returns a~candidate
(or non-empty set of candidates) $f(\elec) \in C$ (or $f(\elec) \subseteq C$) given an election~$\elec$. 
The~latter case --- a voting rule returning a set of candidates --- will be important in this paper because we~explicitly consider voting rules returning multiple tied winners.

We say that voter $v$ ranks candidate $a$ \emph{higher} than candidate $b$ if $a \succ_v b$. 
If $a = b$ or $a \succ_v b$, we~write $a \succeq_v b$ and say that $v$ ranks $a$ \emph{weakly higher} than $b$. 
We also extend this notation to \emph{coalitions}, i.e., non-empty subsets of voters. We write $a \succ_T b$ to denote that every voter in coalition $T$ ranks candidate $a$ higher than candidate $b$. 
Extending this notation further, we also write $A \succ_T b$ to denote that $a \succ_T b$ for all $a \in A$.
 
We refer to the candidate ranked highest by voter $v$ as the \emph{top choice} of $v$, denoted by $\top(v)$. 
The~number of voters who have candidate $c$ as their top choice is referred to as the \emph{plurality score} of~$c$, denoted by $\plu(c)$.  
The candidate ranked lowest by voter $v$ is referred to as the \emph{bottom choice} of~$v$, denoted by $\bot(c)$. 
The~number of voters who have candidate $c$ as their bottom choice is referred to as the \emph{veto score} of $c$, denoted by $\veto(c)$.  
We also use $\top_A(v)$ and $\bot_A(v)$, respectively, to denote the top and bottom choices of voter $v$ among the subset $A$ of candidates.

\subsection{Metric Distortion} \label{sec:distortion}

A \emph{metric} over a set $S$ is a function $d : S \times S \rightarrow \mathbb{R}_{\geq 0}$ which satisfies the following conditions for all $a, b, c \in S$: 
(1)~Positive Definiteness: $d(a,b) = 0$ if and only\footnote{Our proofs do not require the ``only if'' condition, so technically, all our results hold for pseudo-metrics, not just metrics.}  if $a=b$; 
(2)~Symmetry: $d(a, b) = d(b, a)$; 
(3)~Triangle inequality: $d(a,b) + d(b,c) \geq d(a,c)$.
  
Given an election \election, we say that a metric $d$ over%
\footnote{We only care about the distances between voters and candidates, so $d$ can be defined as a function $d: V \times C \to \R_{\geq 0}$ instead of on $V \cup C$. The triangle inequality can then be written as $0 \leq d(v, c) \leq d(v, c') + d(v', c') + d(v', c)$ for all $v, v' \in V$ and for all $c, c' \in C$.} 
$V \cup C$ is \emph{consistent} with the ranking $\succ_v$ of voter $v$ if $d(v, a) \leq d(v, b)$ for all $a, b \in C$ such that $a \succ_v b$. 
We say that $d$ is consistent with a ranked-choice profile $\prof$ if it is consistent with the ranking $\succ_v$ for all $v \in V$. 
We use $\dom(\prof)$ to denote the domain of metrics consistent with $\prof$.

The \emph{(utilitarian) social cost} of a candidate $c$ under a metric $d$ is defined as the total distance of voters to $c$, i.e., $\cost(c, d) = \sum_{v \in V} d(v, c)$. 
A candidate $c^*_d $ is \emph{optimal} with respect to the metric $d$ if $c^*_d \in \argmin_{c \in C} \cost(c, d)$.  

The \emph{distortion of a candidate $c$} in an election $\elec$, denoted by $\dist_\elec(c)$, is the~largest possible ratio between the cost of  $c$ and that of an optimal candidate $c^*_d$ with respect to a metric $d \in \dom(\prof)$. That~is,
\[\dist_\elec(c) = \sup_{d \in \dom(\prof)} \frac{\cost(c, d)}{\cost(c^*_d, d)}.\]

The \emph{distortion of a voting rule $f$}, denoted by $\dist(f)$, is the worst-case distortion of any winner of $f$ among all elections. That is,
\[\dist(f) = \max_{\elec} \max_{c \in f(\elec)} \dist_\elec(c).\]

\subsection{Domination Graphs} \label{sec:domination-graphs}

Domination graphs offer a sufficient condition for a candidate to have distortion at most 3, which is known to be the lowest possible distortion that a voting rule can have \citep{anshelevich:bhardwaj:postl}.

\begin{definition}[\emph{Domination Graphs} \citep{gkatzelis:halpern:shah:resolving}]	
Given an election \election, the~\emph{domination graph} of candidate $a$ is the bipartite graph $G_a = (V, C, E_a)$ with the neighborhood function $\neigh_a(v) = \Set{c \in C}{a \succeq_v c}$, i.e., $(v, c) \in E_a$ if and only if $v$ weakly prefers $a$ over $c$.
\end{definition}

Given vectors $\pvec \in \Delta(V)$ and $\qvec \in \Delta(C)$, we refer to a matrix $\mvec \in \Delta(V \times C)$ as a \emph{$\pq$-matching} if $\p{v} = \sum_{c \in C} \m{v, c}$ for all $v \in V$ and $\q{c} = \sum_{v \in V} \m{v, c}$ for all $c \in C$.
We say that candidate $a$ \emph{admits} the $\pq$-matching $\mvec$ if and only if $(v, c) \in E_a$ whenever $\m{v, c} \neq 0$; 
this is equivalent to requiring that $(v,\top_{[\mrow{v} \neq 0]}(v)) \in E_a$ for all $v \in V$.
We say that candidate $a$ is \emph{$\pq$-dominant} if there exists a $\pq$-matching that $a$ admits. 
For any $\pvec \in \Delta(V)$ and $\qvec \in \Delta(C)$, there always exists a~$\pq$-dominant candidate, which was first shown non-constructively by \citet{gkatzelis:halpern:shah:resolving} and then constructively by \citet{PluralityVeto}.

\bigskip

As shown by \citet{gkatzelis:halpern:shah:resolving}, Hall's condition easily generalizes to $\pq$-matchings.
To~state the condition, we extend the neighborhood function to coalitions $T \subseteq V$, by defining that $\neigh_a(T) = \bigcup_{v \in T} \neigh_a(v) = \Set{c \in C}{\text{there exists a } v \in T \text{ such that } a \succeq_v c}$.

\begin{lemma}[\citep{gkatzelis:halpern:shah:resolving}, Lemma 1]
\label{lem:hall}
	For any weight vectors $\pvec \in \Delta(V)$ and $\qvec \in \Delta(C)$, candidate $a$ is $\pq$-dominant if and only if  $p(T) \leq q(\neigh_a(T))$ for all $T \subseteq V$.  
\end{lemma}

\bigskip

Let $p^\uni(v) = 1/n$, for all $v \in V$, be the uniform weight vector.
Let $q^\plu(c) = \plu(c)/m$, for all $c \in C$, be the normalized plurality scores.
We refer to $(\pvec^\uni, \qvec^\mathsf{plu})$-dominant candidates as \emph{plurality dominant candidates}.
The following proposition by \citet{gkatzelis:halpern:shah:resolving} (independently shown by \citet{munagala:wang:improved,DistortionDuality}) captures the main way in which domination graphs are used in the metric distortion framework.

\begin{lemma}[\citep{gkatzelis:halpern:shah:resolving}, Theorem 1]
\label{lem:opt}
In every election, the distortion of every plurality dominant candidate is at most $3$.
\end{lemma}

\section{The Veto Core and its Equivalence to Dominant Candidates}
\label{sec:vetocore}
In this section, we review the classical notion of proportional veto core due to \citet{moulinProportionalVetoPrinciple1981}, then introduce a~generalization called the $\pq$-veto~core and show that it is equivalent to the set of $\pq$-dominant candidates. 

\subsection{The Proportional Veto Core and its Generalization} \label{sec:proportional}

The \emph{proportional veto principle} requires that, for a winning candidate $a$ of an election \election, the following must hold: if $T \subseteq V$ is a~coalition of voters comprising an $\alpha$ fraction of voters (i.e., $\SetCard{T}/n = \alpha)$, then the union of candidates ranked weakly lower than $a$ by voters in $T$ should comprise at least an $\alpha$ fraction of candidates (i.e., $\SetCard{\neigh_a(T)}/m \ge \alpha$).
Intuitively, this provides protection against overly bad candidates, for every possible groups of voters, proportional to their size.
\citet{moulinProportionalVetoPrinciple1981} defined the proportional veto core as the set of candidates not blocked in this way:

\begin{definition}[\emph{Proportional Veto Core} \citep{moulinProportionalVetoPrinciple1981}] \label{def:prop-veto-core}
Given an election $\elec = (V, C, \prof)$, we say that the coalition $T$ \emph{blocks} the candidate $c$ if there exists a subset of candidates $B$ with $B \succ_T c$ and
\[\ceil[\Bigg]{m \cdot \frac{|T|}{n}} - 1 \geq m - |B|.\]
The set of all candidates that are \emph{not} blocked is called the \emph{proportional veto core}.	
\end{definition}

\citet{moulinProportionalVetoPrinciple1981} shows that, for every election, the proportional veto core is non-empty and does~not contain any candidate who is ranked lowest by more than $\frac{n}{m}$ voters.

\medskip

We now generalize the definition to arbitrary weights for voters and candidates.

\begin{definition}[\emph{$\pq$-Veto Core}] \label{def:veto-core}
Given an election $\elec = (V, C, \prof)$ along with weights vectors $\pvec \in \Delta(V)$ and $\qvec \in \Delta(C)$, we say that coalition $T$ $\pq$-\emph{blocks} candidate $c$ if there exists a subset of candidates $B$ with $B \succ_T c$ and $p(T) > 1 - q(B)$. 
We refer to $B$ as a \emph{witness} that $T$ blocks $c$.
The set of all candidates that are \emph{not} $\pq$-blocked is called the~$\pq$-\emph{veto core}.
\end{definition}

We first show that the $\pq$-veto core indeed generalizes the proportional veto core.

\begin{proposition}
	For every election \election, the $(\pvec^\mathsf{uni}, \qvec^\mathsf{uni})$-veto core equals the proportional veto core (where  $p^\mathsf{uni}(v) = 1/n$ for all $v \in V$ and $q^\mathsf{uni}(c) = 1/m$ for all $c \in C$).
\end{proposition}

\begin{proof}
	First, observe that $\ceil[\big]{\frac{x}{y}} - 1 \ge z$ if and only if $\frac{x}{y} > z$ for all natural numbers $x, y, z \in \N$.
	Candidate $c$ is is blocked (in the traditional sense) if and only if there exists a coalition $T$ and a subset of candidates $B$ such that $B \succ_T c$ and $\ceil[\Big]{m \cdot \frac{|T|}{n}} - 1 \geq m - |B|$.
	Because all quantities are integers, this holds if and only if $m \cdot \frac{|T|}{n} > m - |B|$, which in turn holds if and only if $c$ is $(\pvec^\mathsf{uni}, \qvec^\mathsf{uni})$-blocked by $T$ with witness $B$.
\end{proof}

\subsection{Relation to Dominant Candidates} \label{sec:equivalance}

We now show that the $\pq$-veto core equals the set of $\pq$-dominant candidates. 

\begin{theorem} \label{thm:eq}
  For every election \election and for any weight vectors $\pvec \in \Delta(V)$ and $\qvec \in \Delta(C)$, the $\pq$-veto core equals the set of $\pq$-dominant candidates.
\end{theorem}

\begin{proof} 
	For any election $\elec = (V, C, \prof)$ and any weight vectors $\pvec \in \Delta(V)$ and $\qvec \in \Delta(C)$, we show that candidate $a$ is not $\pq$-dominant if and only if $a$ is not in the $\pq$-veto core of $\elec$.

	First, assume that $a$ is not $\pq$-dominant.
	Then, by Hall's condition (\cref{lem:hall}), there exists a~coalition $T$ of voters such that $p(T) > q(\neigh_a(T)) = 1 - q(C \setminus \neigh_a(T))$. 
	Since $C \setminus \neigh_a(T) = \Set{c \in C}{c \succ_T a}$ by definition, $C \setminus \neigh_a(T) \succ_T a$.
	Hence, coalition $T$ $\pq$-blocks $a$ with witness $C \setminus \neigh_a(T)$. 
	Thus, we have shown that $a$ is not in the  $\pq$-veto core.
	
	For the converse direction, assume that candidate $a$ is not in the $\pq$-veto core, i.e., there exists a coalition $T$ that $\pq$-blocks $a$.
	Let $B$ be a witness for $T$ $\pq$-blocking $a$, i.e., $B \succ_T a$ and $p(T) > 1 - q(B)$.
        Because $B \subseteq \Set{c \in C}{c \succ_T a} = C \setminus \neigh_a(T)$, we can bound its weight by $q(B) \leq q(C \setminus \neigh_a(T))$.
	This implies that Hall's condition (\cref{lem:hall}) is violated for the coalition $T$ as follows:
	$p(T) > 1-q(B) \geq 1- q((C \setminus \neigh_a(T))) = q(\neigh_a(T))$.
  	Thus, we have shown that $a$ is not $\pq$-dominant. 
\end{proof}

\begin{corollary} \label{cor:witness}
A candidate $c$ is in the $\pq$-veto core if and only if $c$ admits a $\pq$-matching. 
\end{corollary}

\section{A Class of Voting Rules with Ties Characterizing The Veto Core}
\label{sec:ties}
In this section, we show that the $\pq$-veto core is exactly characterized by the set of candidates who can be elected by a natural iterative elimination process generalizing \pluralityveto.
Defining this generalized algorithm correctly requires us to remedy flaws in \pluralityveto related to tie breaking.
These flaws are most clearly observed in two examples: 
(1) In \cref{ex:obvious-tie}, only a single winner --- depending on the veto order --- is chosen, even though the election should clearly be declared a tie. 
(2) In \cref{ex:resolvability}, \pluralityveto never elects candidate $b$ since $\plu(b)=0$, even though $b$ is a winner of \pluralitymatching (i.e., a plurality dominant candidate);
thus, candidates are not treated equally under the sufficient condition of achieving distortion 3 (given in \cref{lem:opt}).
We propose a generalized version of \pluralityveto with a less aggressive elimination criterion, which will remedy all of these flaws simultaneously.

Before doing so, we approach the issue of tied winners from a different angle. 
\cref{cor:witness} indicates that membership of a candidate $a$ in the $\pq$-veto core can be certified by a $\pq$-matching \mvec that $a$ admits.
If we view $\mvec$ as a ``justification'' for electing $a$, then this justification equally suggests electing any other candidate $c$ admitting $\mvec$.
This view motivates the following definition.

\begin{definition}[Ties]
  Given a $\pq$-matching $\mvec$, the set $W(\mvec) = \Set{c \in C}{c \text{ admits } \mvec}$ is called the set of \emph{winners tied at} $\mvec$.
\end{definition}

\subsection{Ties are Convex} \label{sec:convex}

The \emph{combinatorial} definition of ties based on matchings turns out to imply a natural geometric condition on the candidates' embedding.
We show that for every embedding of the candidates (consistent with the given rankings) into a~metric space allowing the definition of convexity, $W(\mvec)$ is convex for all $\pq$-matchings $\mvec$.
This is desirable for the goal of electing a central candidate: if we think of the metric space as capturing political positions, then convexity says that if more extreme candidates \emph{in all directions} can be elected, a~candidate who is moderate in the sense of interpolating their positions should also be elected.
Convexity in all metric spaces consistent with the given rankings is implied by the following property:

\begin{definition}[Prefix-Intersection]
  Given an election \election, a subset of candidates $A$ is \emph{\uc} if
  there exist indices $k_v \in [0, m]$ for all voters $v$ such that $A$ is the intersection of the sets of top $k_v$ candidates for each voter $v \in V$.
\end{definition}

We first show that \uc sets of candidates are convex for any metric embedding consistent with the given rankings.

\begin{proposition} \label{prop:convexity}
  Given an election \election, if $A$ is a \uc set of candidates, then for every metric space consistent with $\prof$ for which convexity is defined, $A$ is convex, in the sense that the~interior of the convex hull of $A$ contains no point from $C \setminus A$.
\end{proposition}

\begin{proof}
  For each voter $v$, the set of the top $k_v$ candidates forms a ball in any metric space --- except that some of the candidates on the boundary of the ball may be included while others are excluded (due to indifferences in the distance-based ranking which are resolved arbitrarily in the~reported rankings).
  The interior of each ball is a convex set, and thus, so is the intersection of the~interiors of the balls.
\end{proof}

We now show that the set of winners tied at any $\pq$-matching is \uc.

\begin{proposition} \label{prop:ties}
  For every $\pq$-matching $\mvec$, the set $W(\mvec)$ of winners tied at $\mvec$ is \uc.
\end{proposition}

\begin{proof}
	By definition (see \cref{sec:domination-graphs} for this characterization), the set of winners tied at $\mvec$ is $W(\mvec) = \Set{c \in C}{c \succeq_v \top_{[\mrow{v} \neq 0]}(v) \text{ for all } v \in V}$.
	For every voter $v$, note that the set $T_v = \Set{c \in C}{c \succeq \top_{[\mrow{v} \neq 0]}(v)}$ is the set of top $k_v$ choices for some $k_v \in [1, m]$. 
	Therefore, $W(\mvec) = \bigcap_{v \in V} T_v$ is \uc by definition.
\end{proof}

\citet{PluralityVeto} conjectured that the $\pq$-veto core is convex in all metric spaces consistent with the rankings, i.e., not only $W(\mvec)$ but also the union of all the $W(\mvec)$.\footnote{The conjecture was originally phrased only for the special case of the \emph{Peer Selection} setting, in which the set of voters is the~same as the set of candidates (due to the aforementioned tie breaking issues with \pluralityveto that will be resolved in \cref{sec:serialveto,sec:characterization}).}
We resolve this conjecture negatively in \cref{ex:convexity}.

\begin{figure}[t]
	\includegraphics[scale=0.25]{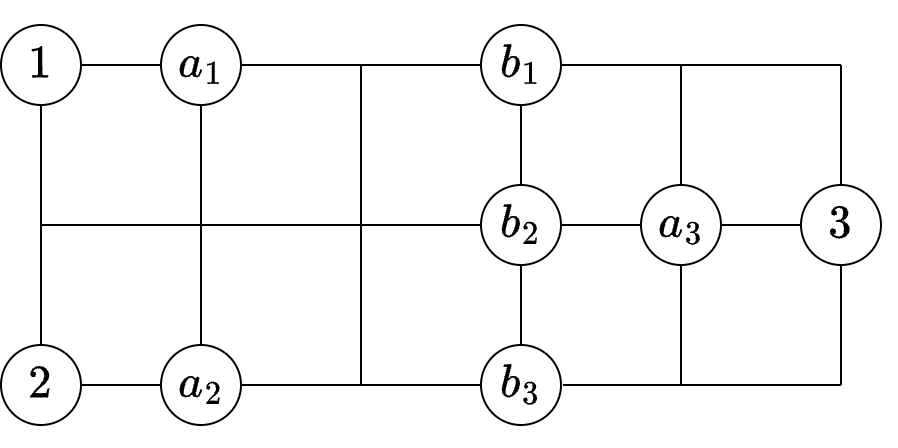}	
	\caption{A consistent embedding of the election in \cref{ex:convexity} into a 2-dimensional $\ell_1$-normed space such that the plurality veto core is non-convex.}
	\label{fig:convexity}
\end{figure}

\begin{example} \label{ex:convexity}
	There are three voters $V = \SET{1, 2, 3}$ and six candidates $C = \SET{a_1, a_2, a_3, b_1, b_2, b_3}$.  
	The~rankings are as follows: 
	\begin{itemize}
		\item $a_1 \succ_1 b_1 \succ_1 a_2 \succ_1 b_2 \succ_1 b_3 \succ_1 a_3$
		\item $a_2 \succ_2 b_3 \succ_2 a_1 \succ_2 b_2 \succ_2 b_1 \succ_2 a_3$
		\item $a_3 \succ_3 b_2 \succ_3 b_1 \succ_3 b_3 \succ_3 a_1 \succ_3 a_2$
	\end{itemize}
	Observe that $b_1$ and $b_3$ respectively admit the plurality matchings $\SET{(1, a_2), (2, a_3), (3, a_1)}$ and $\SET{(1, a_3), (2, a_1), (3, a_2)}$, while $b_2$ does not admit any plurality matching, as voters $1$ and $2$ can be only matched to $a_3$ but $\plu(a_3) = 1$. 
	Thus, $b_1$ and $b_3$ are in the plurality veto core while $b_2$ is not.
	In~\cref{fig:convexity}, we illustrate a 2-dimensional $\ell_1$-normed space consistent with these rankings where $b_2$ is a~convex combination of $b_1$ and $b_3$, specifically $b_2 = \frac{b_1 + b_3}{2}$.
	This example disproves the~conjecture that the plurality veto core is always convex.\footnote{While \cref{ex:convexity} disproves the conjecture in general, it remains open for the special case of the peer selection setting.} 
\end{example}

\subsection{\large \sequentialrule} \label{sec:serialveto}

We now present an adaptation of \pluralityveto (for the general case with arbitrary weights) with more careful elimination.
We present this class of rules only for integral weight vectors (instead of weights normalized to 1); when $\pvec, \qvec$ are rational, they can of course be made integral by multiplying all weights.
For irrational weights, the rules can be adapted, albeit at an increase in notational complexity. The corresponding results are also subsumed by our continuous-time approach in \cref{sec:anonrule}.
Given a set $S$ and a positive integer $N$, let $\Delta_N(S)$ denote the domain of integral weight vectors over $S$ that add up $N$.

\begin{definition}[Veto Order]
	Given an election \election along with a weight vector $\pvec \in \Delta_N(V)$, a~\emph{veto order} under $\pvec$ is a repeated sequence of voters $(v_1, \ldots, v_N)$ in which each voter $v \in V$ occurs exactly $\p{v}$ times.
\end{definition}

\sequentialrule with veto order $(v_1, \ldots, v_N)$, given in \cref{alg:sequential}, starts off with the set of all candidates as eligible winners.
Then, the rule processes voters in the order $(v_1, \ldots, v_N)$.
In each iteration, the voter being processed starts from her bottom choice among not-yet-eliminated candidates and eliminates those whose weight is already 0.  
When she arrives at a candidate $c$ with positive weight, she decrements his weight by 1, \emph{but does not eliminate $c$, even if his weight is now 0}.  
The set of winners consists of all candidates that are not eliminated until the end. 

\begin{theorem} \label{thm:sequential-outcome}
	\sequentialrule (\cref{alg:sequential}) returns a non-empty subset of candidates $W$ consisting of the~winners tied at a $\pq$-matching.
\end{theorem}

\begin{proof}
	The proof is similar to that of Theorem~1 and Theorem~2 in \citet{PluralityVeto}.
	We first show that \cref{alg:sequential} returns a non-empty subset of candidates.
	The weights of candidates initially add up to $N$. 
	In each of the $N$ iterations of the \textbf{for} loop (lines \ref{alg:for-start}--\ref{alg:for-end}), a positive weight is decremented by 1.
	Hence, there is a candidate $w$ whose weight reaches 0 only at the end of the $N^{\text{th}}$ iteration.
	Since a candidate can be removed from $W$ only after his weight reaches 0, $w \in W$ when the algorithm terminates.
	Thus, \cref{alg:sequential} returns a non-empty subset of candidates. 

	Next, we show that $W$ is the set of winners tied at some $\pq$-matching when the algorithm terminates.
	Let $c_i$ denote the candidate whose weight is decremented in the $i^{\text{th}}$ iteration of the \textbf{for} loop (lines \ref{alg:for-start}--\ref{alg:for-end}).
	Let $\mvec \in \Delta_N(V \times C)$ be defined by
	$\m{v, c} = |\Set{ i \in [1, N]}{ v = v_i \text{ and } c = c_i}|$ 
	for all $v \in V$ and $c \in C$. 
	Note that $\mvec$ is a $\pq$-matching since each voter $v$ occurs $\p{v}$ times in the veto order and the weight of each candidate $c$ (initialized with $\q{c}$) reaches 0 at the end.
	We will show that \cref{alg:sequential} returns the winners tied at $\mvec$.
	
	For all $i \in [1, N]$, let $C_i$ denote the candidates that voter $v_i$ ranks weakly higher than $c_i$, i.e., $C_i = \Set{c \in C}{c \succeq_{v_i} c_i}$. 
	All candidates are initially in $W$, and for all $i \in [1, N]$, in the $i^{\text{th}}$ iteration of the \textbf{for} loop (lines \ref{alg:for-start}--\ref{alg:for-end}), candidates that $v_i$ ranks lower than $c_i$ (i.e., candidates that are not in $C_i$) are removed by the \textbf{while} loop (lines \ref{alg:while-start}--\ref{alg:while-end}). 
	By induction on the iterations, $W = \bigcap_{i=1}^N C_i$ at the end.
	By definition, the set of winners tied at $\mvec$ is 
	$W(\mvec) = \Set{c \in C}{c \succeq_{v} \top_{[\mvec(v) \neq 0]}(v) \text{ for all } v \in V} = \Set{c \in C}{c \succeq_{v_i} c_i \text{ for all } i \in [1, N]} = \bigcap_{i=1}^N C_i$, which completes the proof.
\end{proof}

\begin{algorithm}[t]
\caption{\sequentialrule}
\label{alg:sequential}

\begin{algorithmic}[1]

\Statex{\: \ \textbf{Input:}\: An election $\elec = (V, C, \prof)$,}  
\Statex{\phantom{\: \ \textbf{Input:}\:} an integral weight vector $\qvec \in \Delta_N(C)$,}
\Statex{\phantom{\: \  \textbf{Input:}\:} a veto order $(v_1, \ldots, v_N)$ under $\pvec \in \Delta_N(V)$.}

\medskip
\hrule
\bigskip

\Statex{\textbf{Output:}\: A non-empty subset of candidates $W$ in the $\pq$-veto core of $\elec$}
\Statex{\phantom{\textbf{Output:}\:} who are the winners tied at a $\pq$-matching $\mvec$, i.e., $W = W(\mvec)$.}

\medskip
\hrule
\bigskip

\State{\textbf{initialize} $W = C$} 
\State{\textbf{initialize} $\weight(c) = \q{c}$ for all $c \in C$}

\medskip

\For{$i = 1, \ldots, N$} \label{alg:for-start}
	\While{$\weight(\bot_W(v_i)) = 0$} \label{alg:while-start}
			\State{\textbf{remove} $\bot_W(v_i)$ \textbf{from} $W$} 
	\EndWhile \label{alg:while-end}
	\State{\textbf{decrement} $\weight(\bot_W(v_i))$ \textbf{by} $1$} 
\EndFor \label{alg:for-end}

\smallskip
	
\State{\textbf{return} $W$}
\end{algorithmic}
\end{algorithm}

For unit weights over voters and plurality scores over candidates, notice that a veto order corresponds to an ordering of voters and \sequentialrule (\cref{alg:sequential}) returns a set of winners containing the winner of \pluralityveto for the same ordering.
Naturally, in this special case, every possible winner is guaranteed to have distortion 3 by \cref{lem:opt}.

Unlike \pluralityveto, the voting rule proposed by \citet{moulinProportionalVetoPrinciple1981} for electing a candidate from the veto core does not insist on picking a single winner.
However, while it does allow for ties, it~frequently fails to include multiple candidates who should be ``obviously'' tied; in particular, it~does~not always return candidates tied at a common matching.
Indeed, the ties allowed by Moulin's rule appear to be less by design\footnote{For example, the rule cannot return tied candidates when $n$ and $m$ are relatively prime.} than an artifact of the combination of two choices: (1)~eliminating a~candidate when his score reaches 0, and (2) insisting that the winner not be eliminated at all (rather than only at the end). 
 We believe that the combination of these two choices is the~primary reason why \cref{def:prop-veto-core} is a bit more convoluted than \cref{def:veto-core}, and it also introduces various number-theoretic arguments into the analysis of the rule.
  
\subsection{Characterization of the Veto Core via \sequentialrule} \label{sec:characterization}

\Cref{thm:sequential-outcome} guarantees that for every veto order $(v_1, \ldots, v_N)$, the outcome of running \sequentialrule for this veto order results in a set of winners tied at some $\pq$-matching.
Here, we show the~converse: for every $\pq$-matching $\mvec$, there exists a veto order under which the set of winners is $W \supseteq W(\mvec)$
In particular, this implies that the $\pq$-veto core is exactly characterized by the set of possible winners under \sequentialrule with suitably chosen veto orders.

\medskip

The main technical tool in our proof is the concept of a bottom trading cycle, similar to the~notion of top trading cycles widely used in fair allocation tasks.

\begin{definition}[Bottom Trading Cycles]
	Given an election \election along with weight vectors $\pvec \in \Delta_N(V)$ and $\qvec \in \Delta_N(C)$, a \emph{bottom trading cycle} of a $\pq$-matching $\mvec \in \Delta_N(V \times C)$ is an~alternating sequence of voters and candidates $\phi = (v_1, c_1, \ldots, v_k, c_k)$ such that
	\begin{itemize}
		\item $\m{v_i, c_i} > 0$ for all $i \in [1, k]$,
		\item $\bot_{[\qvec \neq 0]}(v_i) = c_{i+1}$ for all $i \in [1, k-1]$, and $\bot_{[\qvec \neq 0]}(v_k) = c_1$.
        \end{itemize}
                
	We use $\mhatvec = \mvec \circ \phi$ to denote the outcome of \emph{swapping along $\phi$}, defined as the same (integral) $\pq$-matching as $\mvec$ except
	
	\begin{itemize}
		\item $\mhat{v_i, c_i} = \m{v_i, c_i} - 1$ for all $i \in [1, k]$,
		\item $\mhat{v_i, c_{i+1}} = \m{v_i, c_{i+1}}+1$ for all $i \in [1, k-1]$, and $\mhat{v_k, c_1} = \m{v_k, c_1} + 1$.
	\end{itemize}
\end{definition}

\begin{figure}[t]
\includegraphics[scale=0.27]{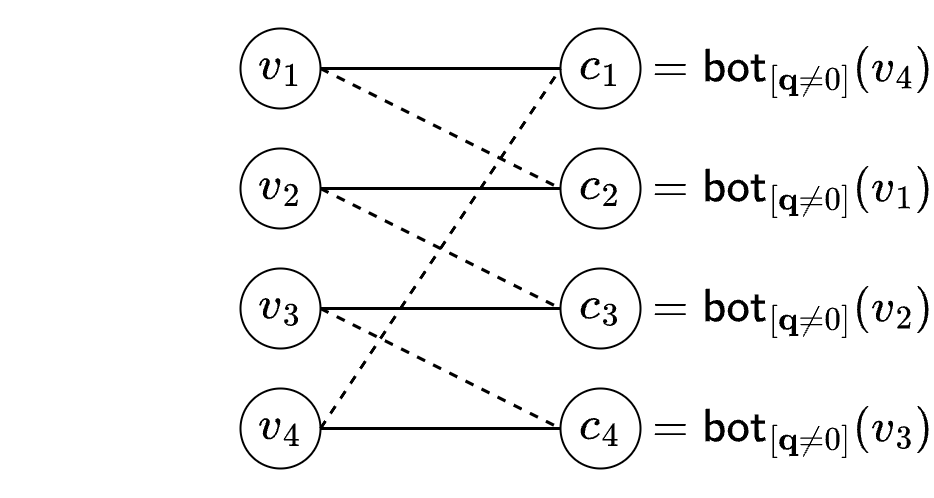}	
\caption{A bottom trading cycle $\phi = (v_1, c_1, \ldots, v_4, c_4)$ of a $\pq$-matching $\mvec$ is illustrated above. Notice that $\mvec \circ \phi$ is the same matching as $\mvec$ except that the weight on each bold edge is decremented by 1 and the~weight on each dotted edge is incremented by 1.}
\end{figure}

First, we show that if $\mvec$ is a matching in which \emph{no} voter has positive edge weight to her bottom choice, then $\mvec$ must contain a bottom trading cycle.

\begin{lemma} \label{lem:btc-exists}
	Given an election \election along with weight vectors $\pvec \in \Delta_N(V)$ and $\qvec \in \Delta_N(C)$, if $\mvec \in \Delta_N(V \times C)$ is a $\pq$-matching with $\m{v, \bot_{[\qvec \neq 0]}(v)} = 0$ for all $v \in V$, then $\mvec$ contains a~bottom trading cycle.
\end{lemma}

\begin{proof}
	Suppose that $\m{v, \bot_{[\qvec \neq 0]}(v)} = 0$ for all $v \in V$. 
	If $N=1$ then the lemma holds vacuously.
	Hence, assume that $N > 1$. 
	Then, there exist a voter $v_1$ and a candidate $c_1$ such that $\m{v_1, c_1} \neq 0$.
	For $i = 1, 2, \ldots$, let $c_{i+1} = \bot_{[\qvec \neq 0]}(v_{i})$ and let $v_{i+1}$ be an arbitrary voter such that $\m{v_{i+1}, c_{i+1}} \neq 0$.  
	Note that there always exists such a voter $v_{i+1}$ because $\m{v_{i}, c_{i+1}} = 0$ and $\q{c_{i+1}} > 0$ by definition.
	Since there are only $m$ candidates, after at most $m+1$ steps, a candidate must repeat; let $j \neq k$ be indices such that $c_j = c_k$.
	Then, $(v_j, c_j, v_{j+1}, c_{j+1}, \ldots, v_{k-1}, c_{k-1})$ is a~bottom trading cycle of $\mvec$, which completes the proof.
	\end{proof}

Next, we show that if a candidate admits a matching, and one swaps along a bottom trading cycle, the candidate admits the new matching as well.
      
\begin{lemma} \label{lem:btc-swap}
	Given an election \election along with weight vectors $\pvec \in \Delta_N(V)$ and $\qvec \in \Delta_N(C)$, 
	suppose that $\mvec \in \Delta_N(V \times C)$ is a $\pq$-matching with bottom trading cycle $\phi=(v_1, c_1, \ldots, v_k, c_k)$.
    If~candidate $a$ admits $\mvec$, then $a$ also admits $\mhatvec = \mvec \circ \phi$.
\end{lemma}

\begin{proof}
	For all $i \in [1, k]$, we show that $(v_i, \bot_{[\qvec \neq 0]}(v_i)) \in E_a$, which suffices to prove that candidate $a$ admits $\mhatvec$ since $\mhatvec$ is the same $\pq$-matching as $\mvec$ except that voter $v_i$ reassigns a~unit of her weight from candidate $c_i$ to $\bot_{[\qvec \neq 0]}(v_i)$.  
 	
 	Let $i \in [1, k]$.
 	Since candidate $a$ admits $\mvec$ and $\m{v_i, c_i} \neq 0$, we have $(v_i, c_i) \in E_a$, i.e., $a \succeq_{v_i} c_i$.  
	Note that $\q{c_i} \neq 0$ since $\m{v_i, c_i} \neq 0$. 
	Then, $c_i \succeq_{v_i} \bot_{[\qvec \neq 0]}(v_i)$ by definition.
	Combining the~two, we get that $a \succeq_{v_i} \bot_{[\qvec \neq 0]}(v_i)$.
	Thus, $(v_i, \bot_{[\qvec \neq 0]}(v_i)) \in E_a$. 
	This completes the proof.
\end{proof}

\begin{theorem} \label{thm:core-characterization-veto-order}
  	Let \election be an election with weight vectors $\pvec \in \Delta_N(V)$ and $\qvec \in \Delta_N(C)$, and let $\mvec \in \Delta_N(V \times C)$ be an (integral) $\pq$-matching. 
  	There exists a veto order $\sigma = (v_1, \ldots, v_N)$ under~$\pvec$ for which \sequentialrule (\cref{alg:sequential}) returns a subset of candidates $W$ with $W \supseteq W(\mvec)$.
\end{theorem}

\begin{proof}
    The proof is by induction on $N$.
	As the base case, consider the case where $N=0$. 
	Because no iterations are run, all candidates will be winners, establishing the claim.
	Now consider an election \election with $\pvec \in \Delta_N(V)$ and $\qvec \in \Delta_N(C)$ such that $N \geq 1$.     
	Let $\mvec$ be an integral $\pq$-matching, i.e., $\mvec \in \Delta_N(V \times C)$.

    If $\m{v, \bot_{[\qvec \neq 0]}(v)} = 0$ for all voters $v$, then \cref{lem:btc-exists} implies the existence of a bottom trading cycle $\phi$ of $\mvec$.
	Let $\mhatvec = \mvec \circ \phi$.
    By \cref{lem:btc-swap}, we have $W(\mhatvec) \supseteq W(\mvec)$.
    We will construct a veto order $\sigma$ such that $W \supseteq W(\mhatvec) \supseteq W(\mvec)$, which will imply the claim.
    Notice that by definition of bottom trading cycles and $\widehat{\mvec}$, for all voters $v$ along $\phi$, we have $\mhat{v, \bot_{[\qvec \neq 0]}(v)} > 0$.
	In particular, we can update $\mvec$ with $\mhatvec$, and ensure the existence of a voter $v_1$ for whom $\m{v_1, \bot_{[\qvec \neq 0]}(v_1)} > 0$.	
	
	We now define $v_1$ to be the starting voter of the sequence $\sigma$, and let $c_1 = \bot_{[\qvec \neq 0]}(v_1)$.
	Define the~following modified weights $\widebar{\pvec} \in \Delta_{N-1}(V)$ and $\widebar{\qvec} \in \Delta_{N-1}(C)$.
	Let $\widebar{p}(v_1) = \p{v_1}-1$ and $\widebar{p}(v) = \p{v}$ for all $v \in V \setminus \SET{v_1}$, and similarly
	$\widebar{q}(c_1) = \q{c_1}-1$ and $\widebar{q}(c) = \q{c}$ for all $c \in C \setminus \SET{c_1}$.
	Define the modified matching $\mbarvec \in \Delta_{N-1}(V \times C)$ by $\mbar{v_1, c_1} = \m{v_1, c_1} - 1$ and $\mbar{v, c} = \m{v, c}$ for all $(v, c) \in (V \times C) \setminus \SET{(v_1, c_1)}$.	
	Notice that $\mbarvec$ is a $(\widebar{\pvec}, \widebar{\qvec})$-matching.
	By induction hypothesis, there exists a veto order $\widebar{\sigma}=(\widebar{v}_1, \ldots, \widebar{v}_{N-1})$ under $\widebar{\pvec}$ for which \cref{alg:sequential} returns a subset of candidates $W \supseteq W(\mbarvec)$.
	Define $\sigma = (v_1, \widebar{v}_1, \widebar{v}_2, \ldots, \widebar{v}_{N-1})$, which is a veto order under $\pvec$.
	
	Given the input $\elec$, $\qvec$ and $\sigma$, \cref{alg:sequential} first decrements the weight of $c_1$ by 1. 
	Then, the~algorithm runs exactly the same way as if it were given the input $\elec$, $\widebar{\qvec}$ and $\widebar{\sigma}$, and thus, it returns $W$. 
\end{proof}

As an immediate corollary, we obtain the following:

\begin{corollary} \label{cor:core-inside-sequential}
	Every candidate in the $\pq$-veto core can win under \sequentialrule with some veto order.  
\end{corollary}

\begin{proof}
  	Whenever $G_a$ has an integral $\pq$-matching, the claim follows directly from \cref{thm:core-characterization-veto-order}. 
  	But it is impossible that $G_a$ contains only fractional $\pq$-matchings, because all entries of \pvec~and~\qvec are integral, so the existence of an integral matching follows from the integrality part of the Max-Flow Min-Cut Theorem.
\end{proof}


\section{A Practical Voting Rule Electing from the Veto Core}
\label{sec:anonymous}
In this section, we present a continuous-time extension of \sequentialrule called \anonrule which does not require the specification of a veto order; as a result, the rule does not prioritize any voter over another.
Our rule can also be considered as a generalization (beyond the proportional veto core) of the \textsc{VetoByConsumption} algorithm (Definition~12) of \citet{ianovskiComputingProportionalVeto2021}. We will discuss this connection in more detail below.
We~show that this voting rule still picks a~set of winners tied at a $\pq$-matching; thus, it~constitutes a~canonical rule for electing from the~$\pq$-veto core. 
Most importantly, we show that \anonrule with uniform weights over voters and plurality scores over candidates satisfies many axioms desirable in practice.

\subsection{\anonrule} \label{sec:anonrule}

\anonrule, given in \cref{alg:anon}, starts off with the set of all candidates as eligible winners, just like \sequentialrule (\cref{alg:sequential}).
However, instead of a particular veto order, all voters are processed simultaneously in continuous time.
From time 0 to 1, every voter $v$, starting from her bottom choice, \emph{instantly} eliminates those with weight 0 and moves on to her next bottom choice.
When she arrives at a candidate with positive weight, she decrements his weight continuously at rate $p(v)$.
The set of winners consists of all candidates who are not eliminated until time 1.

In \cref{alg:anon}, along with the above continuous-time description, \anonrule is also given using  discrete events corresponding to the time steps at which a candidate is eliminated; this~shows that \anonrule can be implemented efficiently in a standard computational model.
Here, $\bot\inv_W(c)$ denotes the set of voters whose bottom choice among the subset $W$ of candidates is candidate $c$. 

\begin{algorithm}[h]
\caption{\anonrule}
\label{alg:anon}

\algrenewcommand\algorithmicloop{\textbf{from} $\timer=0$ \textbf{to} 1 \algorithmicdo}

\begin{algorithmic}[1]

\Statex{\textbf{Input:}\: An election $\elec = (V, C, \prof)$,}  
\Statex{\phantom{\textbf{Input:}\:} weight vectors $\pvec \in \Delta(V)$ and $\qvec \in \Delta(C)$.}

\medskip
\hrule
\bigskip

\Statex{\textbf{Output:}\: A non-empty subset of candidates $W$ in the $\pq$-veto core of $\elec$}
\Statex{\phantom{\textbf{Output:}\:} who are the winners tied at a $\pq$-matching $\mvec$, i.e., $W = W(\mvec)$.}

\medskip
\hrule
\bigskip

\Statex{\textbf{Initialization:}\:  $W = C$,}
\Statex{\phantom{\textbf{Initialization:}\:} $\weight(c) = q(c)$ for all $c \in C$.}

\medskip
\hrule
\bigskip

\Statex{\large \textit{Continuous-Time Description:}}

\medskip

\Loop 
\ForAll{$v \in V$ \textbf{in parallel}}
	\While{$\weight(\bot_W(v)) = 0$} \label{alg2:while-start}
			\State{\textbf{remove} $\bot_W(v)$ \textbf{from} $W$} 
	\EndWhile \label{alg2:while-end}
	\State{\textbf{decrement} $\weight(\bot_W(v))$ \textbf{at rate} $p(v)$} 
\EndFor
\EndLoop
	
\State{\Return{$W$}}

\bigskip
\hrule
\bigskip

\Statex{\large \textit{Implementation via Discrete Events:}}
\setcounter{ALG@line}{0}

\medskip

\State{$\timer = 0$}
\While{$\timer < 1$}
  \ForAll{$c \in C \setminus [\bm{\weight} \neq 0]$}
     \If{$\invbot_W(c) \neq \emptyset$} 
     	\State{\textbf{remove} $c$ \textbf{from} $W$}
     \EndIf
  \EndFor
  \medskip
  \State{$\delta = \min_{c \in W} \frac{\weight(c)}{p(\invbot_W(c))}$}
  \medskip
  \ForAll{$c \in W$}
    \State{\textbf{decrement} $\weight(c)$ \textbf{by} $\delta \cdot p(\invbot_{W}(c))$}
  \EndFor
  \State{\textbf{increment} $\timer$ \textbf{by} $\delta$}
\EndWhile

\State{\Return{$W$}}
\end{algorithmic}
\end{algorithm}

\begin{theorem} \label{thm:continuous-rule}
	\anonrule (\cref{alg:anon}) returns a non-empty subset of candidates $W$ consisting of the winners tied at a $\pq$-matching.   
\end{theorem}

\begin{proof}
	For proving the theorem, we will be using the continuous-time description of \cref{alg:anon}.
	The proof proceeds along the same lines as that of \cref{thm:sequential-outcome}.
	We first show that \cref{alg:anon} returns a non-empty subset of candidates.
	The weights of candidates initially add up to $1$, and each voter $v \in V$ decrements the weights at rate $\p{v}$ from time 0 to 1. 
	Since $\p{V} = 1$, there is a candidate $w$ whose weight reaches 0 only at time 1. 
	Since a candidate can be removed from $W$ only after his weight reaches 0, $w \in W$ when the algorithm terminates.
	Thus, \cref{alg:anon} returns a~non-empty subset of candidates. 
	
	Next, we show that $W$ consists of the winners tied at some  $\pq$-matching when the algorithm terminates.
	Let $c_v(t)$ denote the candidate whose weight is decremented by voter $v$ at time $t$.
	Let $\ind_{v, c}(t)$ be an indicator function which takes the value 1 if $c = c_v(t)$ and 0 otherwise.
	Let $\mvec \in \Delta_N(V \times C)$ be defined by $\m{v, c} = \int_0^1 \ind_{v, c}(t) \cdot \p{v} \; \dd t$ for all $v \in V$ and $c \in C$, i.e., the total amount by which $v$ decreases the weight of $c$ over the entire algorithm. 
	Note that $\mvec$ is a $\pq$-matching since each voter $v$ decrements the weight of $c_v(t)$ at rate $p(v)$ from time $t=0$ to $1$, and the weight of each candidate $c$ (initialized with $\q{c}$) must have reached 0 by the end.
	We will show that \cref{alg:anon} returns the winners tied at $\mvec$.
	
	We use $C_v(t) = \Set{c \in C}{c \succeq_v c_v(t)}$ to denote the set of candidates that voter $v$ ranks weakly higher than $c_v(t)$. 
	Note that $C_v(t) \subseteq C_v(t')$ for all $t > t'$, because the set of remaining candidates shrinks over time, and $v$ always decrements the weight of her lowest-ranked remaining candidate. 
	All candidates are initially in $W$ and, for~all voters $v \in V$, candidates $c \in W$ that $v$ ranks lower than $c_v(t)$ (i.e., candidates not in $C_v(t)$) are removed at time $t$ by the \textbf{while} loop (lines~\ref{alg2:while-start}~--~\ref{alg2:while-end}). 
	Note that such candidates $c$ must have weight 0 at time $t$; otherwise, by definition, $v$ would decrease the score of such a candidate instead at time $t$.
	By induction on time\footnote{Induction can be formally performed by considering only the time steps when a candidate is removed from $W$, i.e., by considering the discrete set of values that are assigned to the variable $\timer$ in the implementation via discrete events.},  $W = \bigcap_{t=0}^1 \bigcap_{v \in V}  C_v(t) = \bigcap_{v \in V} C_v(1)$ at the~end.
	By definition, the set of winners tied at $\mvec$ is $W(\mvec) = \Set{c \in C}{c \succeq_{v} \top_{[\mvec(v) \neq 0]}(v) \text{ for all } v \in V} 
			 = \Set{c \in C}{c \succeq_{v} c_v(1) \text{ for all } v \in V}
			 = \bigcap_{v\in V} C_v(1)$
	which completes the proof.
\end{proof}

As mentioned above, \anonrule is a generalization of the \textsc{VetoByConsumption} rule of \citet{ianovskiComputingProportionalVeto2021} picking a candidate from the proportional veto core.
Therefore, their rule also satisfies our definition of ties, and consequently, the convexity property given in \cref{sec:convex}.
However, the generalization to arbitrary weights, while preserving the~convexity property, is not trivial because if the weight of a candidate $c$ is 0 and no voter ranks $c$ at the bottom, then the~naive generalization is not well-defined, as it involves dividing 0 by 0. 
Another way to deal with this issue would be to ignore candidates with weight 0 as in \pluralityveto; however, doing so may lead to violating convexity of the set of returned winners.

We also remark that \anonrule can be easily generalized further. 
The rates at which a voter decrements the weight of her bottom choice need not stay constant over time. 
Instead, a~voting rule can prescribe rates $r_v(t) \geq 0$ for all voters $v$ and times $t \in [0,1]$, so long as these rates satisfy $\int_0^1 r_v(t) \dd t = p(v)$. 
The proof of \cref{thm:continuous-rule} then carries through verbatim. 
This general rule then also subsumes the discrete rule \sequentialrule, by having voter $v_i$ decrement at rate $N$ during the interval $[(i-1)/N, i/N)$, while all other voters have rate 0 during that interval.

\subsection{Axiomatic Properties of \anonplurule} \label{sec:axioms}

We refer to the instantiation of \anonrule (\cref{alg:anon}) with uniform weights over voters and plurality scores over candidates as \anonplurule.
For ease of notation, in the analysis in this section, we treat \anonrule as using \emph{unnormalized} weights, i.e., the~initial scores are $q(c) = \plu(c)$, and each voter reduces scores at a rate of 1. Of course, scaling all $p(v)$ and $q(c)$ by $1/n$ translates the arguments to an implementation with normalized scores and rates.
Our main result here is that in addition to guaranteeing metric distortion 3 (due to \cref{lem:opt}), \anonplurule satisfies a number of desirable axioms:

\begin{itemize}
\item \axiom{Anonymity} and \axiom{Neutrality}, respectively, require the voting rule to not discriminate a priori between voters and candidates.
In other words, the outcome should depend only on the rankings, not the identities of voters or candidates.

\item \axiom{Resolvability} requires that for any set of tied winners $W$, adding a new voter $v$ whose top choice is some candidate $w \in W$ makes $w$ the unique winner. 

\item \axiom{Monotonicity} requires that when a voter moves a winning candidate $w$ higher up in her ranking, then $w$ still wins. 

\item \axiom{Mutual Majority} requires that if a strict majority of voters rank all candidates in $S$ ahead of all other candidates, then the winner(s) must be from $S$. 
\axiom{Majority} requires this only for singleton sets.
\axiom{Majority Loser} requires this only for sets whose complement is a singleton. 

\item \axiom{Reversal Symmetry} requires that if a candidate $w$ is the unique winner, and the~ranking of each voter is reversed, then $w$ must not be a winner any longer.
\end{itemize}

\begin{theorem}\label{thm:axioms}
	\anonplurule returns a non-empty set $W$ of tied winners  all of whom have distortion at most 3, and for every metric space consistent with the rankings for which convexity is defined, $W$ is convex\footnote{In the sense that the interior of the convex hull of $W$ contains no candidates other than those in $W$ (see \cref{sec:convex}).}.
	Furthermore, the voting rule satisfies \axiom{Anonymity}, \axiom{Neutrality}, \axiom{Resolvability}, \axiom{Monotonicity}, \axiom{Majority}, \axiom{Majority Loser}, \axiom{Mutual Majority}, and \axiom{Reversal Symmetry}.   
\end{theorem}

\subsection{Axioms Violated by \anonplurule}

Before giving the proof of \cref{thm:axioms}, we discuss some other important axioms that are violated by \anonplurule.

Candidate $a$ is a \emph{Condorcet winner} if for every other candidate $b$, there are at least $n / 2$ voters who rank $a$ higher than $b$. 
A Condorcet winner does not always exist. 
A voting rule satisfies \axiom{Condorcet Consistency} if it picks a Condorcet winner whenever one exists. 
\citet{gkatzelis:halpern:shah:resolving} showed that a Condorcet winner is not always a plurality dominant candidate; thus, neither \pluralitymatching nor \pluralityveto nor \anonplurule is Condorcet consistent. 
On the other hand, \citet{anshelevich:bhardwaj:postl} showed that the distortion of any Condorcet winner is at most 3; as a~result, all of the above voting rules can be made Condorcet consistent by forcing them to pick a~Condorcet winner when one exists, without hurting the distortion. 

A voting rule $f$ satisfies \axiom{Electoral Consistency} \cite{youngSocialChoiceScoring1975a} if, when the election \election is partitioned into two elections $\elec_1 = (V_1, C, \prof)$ and $\elec_2 = (V_2, C, \prof)$ with disjoint voters such that $f(\elec_1) \cap f(\elec_2) \neq \emptyset$, 
the set of winners of the whole election is $f(\elec) = f(\elec_1) \cap f(\elec_2)$.
The following example shows that $f=\;$\anonplurule violates \axiom{Electoral Consistency}.

\begin{example}\label{ex:consistency}
There are five candidates $C = \SET{a,b,c,d,e}$, and five voters partitioned as $V_1 = \SET{1, 2, 3}$ and $V_2 =\SET{4, 5}$. 
The rankings are as follows: 
\begin{multicols}{2}
	\begin{itemize}
		\item $a \succ_1 c \succ_1 d \succ_1 b \succ_1 e$
		\item $b \succ_2 c \succ_2 a \succ_2 d \succ_2 e$
		\item $d \succ_3 c \succ_3 b \succ_3 a \succ_3 e$
	\end{itemize}
	\columnbreak
	\begin{itemize}
		\item $c \succ_4 a \succ_4 b \succ_4 d \succ_4 e$
		\item $e \succ_5 d \succ_5 b \succ_5 a \succ_5 c$
	\end{itemize}
\end{multicols}
\noindent $f(\elec) \neq f(\elec_1) \cap f(\elec_2)$ because $f(\elec) = \SET{a, b, c}$, $f(\elec_1) = \SET{a, b, c, d}$ and $f(\elec_2) = \SET{a, b, c, d, e}$.
\end{example}

Perhaps most concerning is that \anonplurule violates the \axiom{Pareto Property}, which requires that a Pareto dominated candidate $a$ (i.e., a candidate whom \emph{every} voter ranks below some candidate $b$) can never be chosen as a winner.
The Pareto property is violated in an input with $m$ candidates and two voters with respective rankings $c_1~\succ~c_2~\succ~c_3 \succ \cdots \succ c_{m-1} \succ c_m$ and $c_m \succ c_2 \succ c_3 \succ \cdots \succ c_{m-1} \succ c_1$.
In this input, \emph{all} candidates are tied at the unique matching, and hence are returned as winners by \anonplurule, even though $c_3, \ldots, c_{m-1}$ are Pareto dominated by $c_2$.
Fortunately, in a sense, this type of preferences is the only setting in which a Pareto dominated candidate can win: a Pareto dominated candidate $a$ can only win if the~candidate $b$ dominating him (and hence $a$ himself) has plurality score 0. 
To~see~this, notice first that the score of $b$ cannot be decreased until $a$ has been eliminated.
Because the score of $b$ must reach 0 by time 1, $b$'s score must be decremented before time 1, so $a$ would be eliminated before time 1, and could not win.

To restore the \axiom{Pareto Property}, one could of course post-process the set of winners $W$ and eliminate any Pareto dominated candidates.
This will leave a non-empty set of winners, because with any Pareto dominated candidate $a$, any candidate $b$ dominating $a$ will always be a winner.
Unfortunately, it would come at the expense of violating convexity and the principle that candidates with the same witnessing matching should all be returned by the voting rule.

An alternative approach --- which would resolve the specific example we gave above --- is to have candidates start not with their plurality scores, but with a score that also accounts for intermediate positions, e.g., scores that decrease exponentially with the rank of a candidate, or even Borda count scores.
However, for such scores, the analogue of \cref{lem:opt} is not known to hold, and probably does not hold verbatim.
Thus, this approach will necessitate a more refined analysis of the distortion of \anonrule (or, equivalently, candidates in the \pq-veto core) under more general candidate weights $\qvec$.
We discuss this direction more in \cref{sec:conclusion}.

\subsection{Proof of \cref{thm:axioms}}
The fact that all of the winning candidates have distortion at most 3 holds by \cref{lem:opt}, because the~winners are all in the plurality veto core.
The fact that the set of winning candidates is convex, for every metric space consistent with the rankings for which convexity is defined, follows from \cref{prop:convexity} because \anonrule returns the winners tied at a $\pq$-matching $\mvec$ by \cref{thm:continuous-rule} and $W(\mvec)$ is prefix-intersecting by \cref{prop:ties}.

\anonplurule is completely symmetric with respect to members of both groups, and thus, \axiom{Anonymity} and \axiom{Neutrality} follow directly from the definition.
We prove the~remaining properties as separate lemmas.

\begin{lemma}
  \anonplurule satisfies \axiom{Majority}, \axiom{Majority Loser} \& \axiom{Mutual Majority}.
\end{lemma}

\begin{proof}    
  We prove that \anonplurule satisfies the \axiom{Mutual Majority} property. Then, \axiom{Majority} and \axiom{Majority Loser} are special cases.
  Suppose that a strict majority $V' \subseteq V$ of voters rank all candidates in $S$ ahead of all other candidates.
  Then, the initial total weight of $S$ is strictly greater than that of $\comp{S}$.
  Until all candidates from $\comp{S}$ are eliminated, each voter in $V'$ will veto a candidate in $\comp{S}$, so the total weight on $\comp{S}$ is decreased at a~strictly higher rate than the total weight on $S$. 
  As a result, the total weight of $\comp{S}$ must reach 0 strictly before the total weight of $S$ reaches 0. 
  In particular, no candidate from $\comp{S}$ can be a winner.
\end{proof}

\begin{lemma}
  \anonplurule satisfies \axiom{Reversal Symmetry}.
\end{lemma}

\begin{proof}
  We first prove that if $w$ is the~\emph{unique} winner, then $\plu(w) > \veto(w)$.
  We prove the contrapositive, so we assume that $\veto(w) \geq \plu(v)$. 
  The voters ranking $w$ last, i.e., $\bot\inv(w)$, will oppose $w$ until $w$ is eliminated. 
  Since $w$ wins, this happens only at time 1, so none of the voters in $\bot\inv(w)$ will ever decrement the weight of another candidate. 
  Thus, the total weight of all candidates except $w$ is decremented at rate at most $n-\veto(w) \leq n-\plu(w)$. 
  Since the total weight of all candidates other than $w$ is initially $n-\plu(w)$, it cannot reach 0 until time 1.
  Therefore, at least one candidate other than $w$ is a winner as well.
  So we have proved that $\plu(w) > \veto(w)$. 
  
  In the election with reversed rankings, $\plu(w) < \veto(w)$ instead.
  Then, $w$ is not a plurality-dominant candidate due to Lemma 6 in \citet{gkatzelis:halpern:shah:resolving}; this can also be directly observed by mildly adapting the argument we just gave in the case $\plu(w) \leq \veto(w)$.
  Thus, $w$ cannot be a winner in the election with reversed rankings.	
\end{proof}

The proofs of \axiom{Monotonicity} and \axiom{Resolvability} are more involved. 
They involve comparing the winner(s) under the original rankings with those under suitably modified rankings. 
To facilitate the analysis, we introduce the following notation; the notation with a prime refers to the execution with modified rankings. 

\begin{itemize}
	\setlength\itemsep{0.40em}
	\item $\tau(c)$ and $\tau'(c)$ denote the time when candidate $c$ is eliminated. When $c$ is not eliminated until time 1 (i.e., a winner), we write $\tau(c) = 1$.
	\item For any time $t \in [0,1]$, let $W(t) = \Set{c}{\tau(c) \geq t}$ and $W'(t) = \Set{c}{\tau'(c) \geq t}$ denote the~sets of candidates who are not eliminated before time $t$. Thus, $W = W(1)$ and $W' = W'(1)$ are the winning candidates.
	\item Similarly, let $W_+(t) = \Set{c}{\tau(c) > t}$ and $W'_+(t) = \Set{c}{\tau'(c) > t}$ be the sets of candidates eliminated strictly after time $t$, i.e., the set of non-eliminated candidates at time $t$ after the~\textbf{while} loop (lines \ref{alg2:while-start}--\ref{alg2:while-end}) in the continuous-time description has been executed.  
    \item $\weight(c,t)$ and $\weight'(c,t)$ denote the weight of candidate $c$ at time $t$.
	\item We write $B(c,t) = \invbot_{W(t)}(c)$ and $B'(c,t) = \invbot_{W'(t)}(c)$ for the sets of voters opposing candidate $c$ among candidates not eliminated before time $t$, i.e., the sets of voters whose bottom choice among candidates not eliminated before time $t$ is $c$.
	\item Similarly, we write $B_+(c,t) = \invbot_{W_+(t)}(c)$ and $B'_+ = \invbot_{W_+'(t)}(c)$ for the sets of voters opposing candidate $c$ among candidates eliminated strictly after time $t$.
\end{itemize}

\begin{lemma} \label{lem:monotonicity}
  \anonplurule satisfies \axiom{Monotonicity}.
\end{lemma}

Let $w \in W$ be a winner ranked higher by a voter in the modified election. 
We begin with the following simple proposition, which will be used repeatedly:

\begin{proposition} \label{prop:opposition-monotonicity}
  If $W(t) \supseteq W'(t)$, then $B(c,t) \subseteq B'(c,t)$ for any candidate $c \in W'(t) \setminus \{w\}$.
\end{proposition}

\begin{proof}
  Consider a voter $v \in B(c,t)$. 
  By definition, $v$ ranks $c$ below all candidates $c' \in W(t)$ under the original election.
  Because $c \neq w$, and candidates other than $w$ can only be ranked lower or in the same place in the modified election, $v$ also ranks $c$ below all candidates $c' \in W(t)$ in the~modified election. 
  And because $W'(t) \subseteq W(t)$, $v$ ranks $c$ below all candidates $c' \in W'(t)$ in the~modified election, meaning that $v \in B'(c,t)$.
\end{proof}

We are now ready to prove \cref{lem:monotonicity}.

\begin{extraproof}{\cref{lem:monotonicity}}
Let $t_1 < t_2 < \cdots < t_k = 1$ be the times at which one or more candidate is eliminated under the modified election, i.e., $\SET{t_1, \ldots, t_k} = \Set{\tau'(c)}{c \in C}$.
For notational convenience, we write $t_0 = 0$.

We will prove the following by induction on $i = 0, \ldots, k$:

\begin{itemize}
	\setlength\itemsep{0.40em}
	\item For all $c \in W'(t_i) \setminus \SET{w}$, the weight at time $t_i$ is weakly lower in the modified election, i.e., $\weight'(c,t_i) \leq \weight(c,t_i)$.
	\item $W(t_i) \supseteq W'(t_i)$, i.e., candidates can only be eliminated earlier in the modified instance.
\end{itemize}

For $i=0$, the first part holds because any candidate $c \neq w$ can only have moved down or stayed in the same place in each ranking, so the initial weights are weakly lower.
The second part of the~hypothesis holds because $W'(0) = W(0) = C$. 

For the induction step from $i$ to $i+1$, consider time $t_i$ for $i < k$.
First, consider the \textbf{while} loop (lines \ref{alg2:while-start}--\ref{alg2:while-end}) in the continuous-time description.
By induction over its iterations, in each iteration, we maintain that the set of uneliminated candidates in the original election is a superset of the uneliminated candidates in the modified election.
This holds initially by the outer induction hypothesis.
Whenever it holds, \cref{prop:opposition-monotonicity} implies that $B(c,t_i) \subseteq B'(c,t_i)$ for all $c \neq w$.
In particular, if $c$ gets eliminated at time $t_i$ in the original election, then $B(c,t_i) \neq \emptyset$ and $\weight(c,t_i) = 0$. 
Thus, $B'(c,t_i) \neq \emptyset$, and the \textbf{while} loop also eliminates $c$ at time $t_i$ in the modified election, because $\weight'(c,t_i) \leq \weight(c,t_i) = 0$ and $c$ is opposed by a non-empty set of voters.
Thus, we showed that $W_+(t_i) \supseteq W'_+(t_i)$.

Next, observe that because the sum of all weights at time $t_i$ is exactly $n(1-t_i)$ in both elections, and every candidate $c \neq w$ has $\weight'(c,t_i) \leq \weight(c,t_i)$, we have $\weight'(w,t_i) \geq \weight(w,t_i)$.

Applying \cref{prop:opposition-monotonicity} after the \textbf{while} loop terminates, we obtain that $B'_+(c,t_i) \supseteq B_+(c,t_i)$ for all candidates $c \neq w$.
Since each voter opposes exactly one candidate, $\bigcup_{c \in W_+(t_i)} B_+(c,t_i) = V = \bigcup_{c \in W'_+(t_i)} B'_+(c,t_i)$.
Hence, the facts that $B'_+(c,t_i) \supseteq B_+(c,t_i)$ for all $c \in W'_+(t_i) \setminus \SET{w}$ and $W'_+(t_i) \subseteq W_+(t_i)$ imply that $B'_+(w,t_i) \subseteq B_+(w,t_i)$.
Combining this fact with the bounds on weights, and defining $\delta_c = \frac{\weight(c,t_i)}{|B_+(c,t_i)|}$ and $\delta'_c = \frac{\weight'(c,t_i)}{|B'_+(c,t_i)|}$ (with the definition that $0/0 := \infty$), we observe that $\delta'_w \geq \delta_w$ and $\delta'_c \leq \delta_c$ for all $c \neq w$.

Since the weight of each uneliminated candidate $c$ at time $t_i$ is decreased at rate $|B'_+(c,t_i)|$ under the modified election, the next elimination will happen at time $t_{i+1} = t_i + \min_{c \in W'_{+}(t_i)} \frac{\weight'(c,t_i)}{|B'_+(c,t_i)|}$. 
If $w$ minimizes this expression, then we obtain that $t_{i+1} = t_i + \delta'_w \geq t_i + \delta_w \geq \tau(w) = 1$, and thus, $w$ is a winner.
Otherwise, for all candidates $c \in W_+(t_{i+1}) \setminus \SET{w}$, we have $\weight(c,t_{i+1}) = \weight(c,t_i) - (t_{i+1}-t_i) \cdot |B_+(c,t_i)| \geq \weight'(c,t_i) - (t_{i+1}-t_i) \cdot |B'_+(c,t_i)| = \weight'(c,t_{i+1})$. 
Thus, we showed that the first part of the hypothesis (inequality on weights) holds at time $t_{i+1}$.

A further consequence of this inequality on weights is that, in the time interval $(t_i, t_{i+1})$, no candidate from $W'_+(t_i)$ can be eliminated in the original election, i.e., $W'_+(t_i) \subseteq W(t_{i+1})$. 
Furthermore, $W'(t_{i+1}) = W'_+(t_i)$ by definition of $t_{i+1}$.
Therefore, $W'(t_{i+1}) = W'_+(t_i) \subseteq W(t_{i+1})$, which completes the proof for the second part of the hypothesis.
Thus, we have shown that $w \in W'$. 
\end{extraproof}

\begin{lemma} \label{lem:resolvability}
  \anonplurule satisfies \axiom{Resolvability}.
\end{lemma}

The proof of resolvability is similar to the proof of monotonicity.
Again, let $w \in W$ be a winner in the original election. 
Recall that in the modified election, all voters in the original election keep their original rankings, and there is an additional voter $\hat{v}$ who ranks $w$ at the top; the rest of her ranking is arbitrary.
The following proposition is similar to \cref{prop:opposition-monotonicity}, and captures a stronger inference that can be drawn here.

\begin{proposition} \label{prop:opposition-resolvability}
  If $W(t) \supseteq W'(t)$, then $B(c,t) \subseteq B'(c,t)$ for any candidate $c \in W'(t) \setminus \SET{w}$, and unless $W'(t) = \SET{w}$, there exists a candidate $c \in W'(t) \setminus \SET{w}$ for whom the set inclusion is strict, i.e., $B(c,t) \subsetneq B'(c,t)$. 
\end{proposition}

\begin{proof}
  Fix any candidate $c \in W'(t) \setminus \SET{w}$, and consider a voter $v \in B(c,t)$. 
  By definition, $v$ ranks $c$ below all candidates $c' \in W(t) \supseteq W'(t)$. 
  Therefore, $v$ ranks $c$ below all candidates $c' \in W'(t)$, implying that $B(c,t) \subseteq B'(c,t)$.
  To prove strictness for at least one candidate, assume that $W'(t) \neq \SET{w}$.
  Then, $\hat{v}$ ranks some candidate $c \in W'(t) \setminus \SET{w}$ lowest among $W'(t)$, and thus, $\hat{v} \in B'(c,t) \setminus B(c,t)$ for that candidate $c$.
\end{proof}

\begin{extraproof}{\cref{lem:resolvability}}
As in the proof of \cref{lem:monotonicity}, let $t_1 < t_2 < \cdots < t_k = 1$ be the times at which one or more candidate is eliminated under the modified election, i.e., $\SET{t_1, \ldots, t_k} = \Set{\tau'(c)}{c \in C}$.
For notational convenience, we write $t_0 = 0$.

We first prove the exact same facts about the process by induction on $i = 0, \ldots, k$, namely:

\begin{itemize}
	\setlength\itemsep{0.40em}
	\item For all $c \in W'(t_i) \setminus \SET{w}$, the weight at time $t_i$ is weakly lower in the modified election, i.e., $\weight'(c,t_i) \leq \weight(c,t_i)$.
	\item $W(t_i) \supseteq W'(t_i)$, i.e., candidates can only be eliminated earlier in the modified instance.
\end{itemize}

For $i=0$, the first part of the hypothesis holds because only $w$ has gained one additional point initially from $\hat{v}$, while every other candidate's score is the same as before. The second part holds simply because $W'(0) = W(0) = C$.

For the induction step from $i$ to $i+1$, consider time $t_i$ for $i < k$.
First, an identical proof to that of \cref{lem:monotonicity} proves that $W_+(t_i) \supseteq W'_+(t_i)$, substituting \cref{prop:opposition-resolvability} for \cref{prop:opposition-monotonicity} in the inner inductive step.

In the original election, the sum of all weights at time $t_i$ is $n \cdot (1-t_i)$, while in the modified election, it is $(n+1) \cdot (1-t_i)$. 
Because every candidate $c \neq w$ has $\weight'(c,t_i) \leq \weight(c,t_i)$, we~obtain that $\weight'(w,t_i) \geq 1 + \weight(w,t_i)$.

We now distinguish two cases. 
First, if $W'_+(t_i) = \SET{w}$, then $w$ is the only remaining candidate at time $t_i$, and the time is $t_i < 1$ since $i < k$. 
Thus, $w$ is the unique winner, completing the proof of \axiom{Resolvability}.
Otherwise, we apply \cref{prop:opposition-resolvability}, and obtain that $B'_+(c,t_i) \supseteq B_+(c,t_i)$ for all candidates $c \in W'_+(t_i) \setminus \SET{w}$, and the inclusion is strict for at least one such candidate.
Because each voter opposes exactly one candidate, $\bigcup_{c \in W_+(t_i)} B_+(c,t_i) = V$ and $\bigcup_{c \in W'_+(t_i)} B'_+(c,t_i) = V \cup \SET{\hat{v}}$.
Now, the facts that $B'_+(c,t_i) \supseteq B_+(c,t_i)$ for all $c \in W'_+(t_i) \setminus \SET{w}$ and $W'_+(t_i) \subseteq W_+(t_i)$, along with the fact that the inclusion is strict for at least one candidate, imply that $\SetCard{B'_+(w,t_i)} \leq \SetCard{B_+(w,t_i)}$.
Again, we define $\delta_c$ and $\delta'_c$ as the distance in time by when $c$ would be eliminated at the current opposition rates.
The previous inequalities then imply that $\delta'_c \leq \delta_c$ for all $c \in W'_+(t_i) \setminus \SET{w}$, and $\delta'_w \geq \delta_w$.

The rest of the inductive proof is now exactly the same as for \cref{lem:monotonicity}, proving that under the modified election, the set of eliminated candidates remains a supserset at all times, and the weights remain weakly lower at all times.

Finally, consider time $t_{k-1}$ (i.e., the last step at which a candidate was eliminated).
Assume for contradiction that after the \textbf{while} loop, in the modified election, at least one candidate other than $w$ was uneliminated, i.e., $W'_+(t_{k-1}) \neq \SET{w}$.
Consider the candidate $c$ with $B'_+(c,t_{k-1}) \supsetneq B_+(c,t_{k-1})$ who is guaranteed to exist by the strict inclusion part of \cref{prop:opposition-resolvability}.
This candidate then has $\delta'_c < \delta_c$.
Because no candidate survived until after time 1, we have that $\delta_c \leq 1-t_{k-1}$, and thus, $\delta'_c < 1-t_{k-1}$. This means that $c$ must be eliminated strictly between times $t_{k-1}$ and $1 = t_k$ under the modified election. 
But by definition of the $t_i$, no eliminations occur between $t_{k-1}$ and $t_k$, giving a~contradiction.
Therefore, $W'_+(t_{k-1}) = \SET{w}$, and $w$ is the unique winner under the modified election.
\end{extraproof}


\section{Conclusion and Future Directions}
\label{sec:conclusion}
We introduced a generalization of the veto core, and showed that it is precisely characterized by the winners under a natural class of sequential veto-based voting rules in which the given weights for candidates are counterbalanced by the given weights for voters.
Applying this insight continuously to the veto core for initial candidate weights equal to plurality scores yields a very natural and simple voting rule with metric distortion 3 which remedies the key flaws of the recently proposed \pluralityveto, by being anonymous, neutral, and resolvable.
In addition, this new rule satisfies many important social choice axioms; combined with its simplicity, this makes it well suited for practical adoption.

Our work raises a wealth of interesting questions for future work.
The study of the generalized veto core and its connection to metric distortion is very intriguing.
In particular, for which initial candidate weights \qvec do candidates in the core guarantee constant (or otherwise small) distortion? 
This question is particularly important in the context of eliminating Pareto dominated candidates, which ideally should come without sacrificing optimal distortion or the characterization of winners in terms of matchings.
Which other interesting guarantees are obtained by electing candidates from a weighted core?

An interesting and natural choice of initial weights for candidates are the $k$-approval scores: the~number of voters who rank a given candidate in their top $k$ choices.
Notice that, as $k$ ranges between $1$ and $m$, these initial weights interpolate between the plurality veto core and the proportional veto core.
Thus, they give rise to a spectrum of voting rules achieving different tradeoffs between the protection of minorities and the majority.
For a quantification of majority and minority principles, see \citet{kondratevMeasuringMajorityPower2020}.

The characterization in terms of veto orders also raises interesting questions.
The simultaneous version of the rule was motivated largely by ensuring anonymity, a very desirable fairness criterion for an election.
The focus on fairness raises an interesting question about veto orders: what would the \emph{voters} prefer to be the veto order?
It is not difficult to show that a voter $v$ is always weakly better off (in the sense of preferring the outcome) by moving her vote to the very end of the sequence.
However, we currently do not know whether a stronger version holds: do voters always weakly prefer moving their veto later in the sequence? 
While we have not identified counter-examples to this~stronger conjecture, we suspect that counter-examples might exist.
However, if the conjecture were true, it~would provide additional principled justification for the simultaneous veto process.

We verified that \anonplurule satisfies several desirable axioms for voting rules, but fails to satisfy certain others: the Pareto property, Condorcet consistency, and consistency. 
Over the~years, researchers have identified dozens of desirable axioms for voting rules, and a more complete exploration may be of interest as well --- indeed, \citet{ianovskiComputingProportionalVeto2023} very recently posted an updated version of their [2021] paper which undertakes such an analysis for the~proportional veto core.
For axioms which \anonplurule fails to satisfy, it~would be interesting to see whether modifications to the voting rule can satisfy these axioms, while retaining the distortion-3 guarantee or sacrificing only a small loss in the distortion guarantee.
Perhaps even more fundamentally, it~would be interesting to investigate which combinations of axioms are fundamentally incompatible with low metric distortion.

\anonplurule has the desirable property of \axiom{Resolvability}, meaning that a tie can always be resolved in favor of any winning candidate $w$ by adding a single voter ranking $w$ at the top.
In a sense, this implies that ties should be ``rare''.
The desideratum of rare ties motivates the axiom of \axiom{Asymptotic Resolvability} \citep{schulzeNewMonotonicCloneindependent2011}, requiring that the proportion of profiles resulting in a tie approach zero as the number of voters increases.
It would be of interest to understand whether \anonplurule is asymptotically resolvable, and perhaps even to quantify the probability of ties under a suitable probabilistic model of voter preferences.

\citet{ianovskiComputingProportionalVeto2021} showed that the expected size of the proportional veto core is $m/2$ as $n \rightarrow \infty$ under the ``Impartial Culture'' (uniformly random preferences) assumption.
It would be interesting to study the expected size of the plurality veto core because it gives a lower bound on the expected number of candidates with distortion at most 3 due to \cref{lem:opt}.
It could be also of interest to study the expected size of the $\pq$-veto core in general.

\begin{acks}
We would like to thank Jannik Peters for pointing us to the work of \citet{ianovskiComputingProportionalVeto2021}, and sharing the observation that there is a connection between the proportional veto core and \pluralityveto.
\citet{petersNote2023} contains a writeup of this observation, along with additional connections of \pluralityveto to other notions in social choice theory.
We would also like to thank anonymous reviewers for very detailed and helpful feedback, and Shaddin Dughmi for useful conversations.
\end{acks}

\bibliographystyle{ACM-Reference-Format}
\bibliography{davids-bibliography/names,davids-bibliography/conferences,davids-bibliography/bibliography,davids-bibliography/publications,other_references}

\appendix

\end{document}